\documentclass[11pt]{article}
\usepackage{amsmath,amssymb,amscd,amsthm}
\makeatletter\@addtoreset {equation}{section}\makeatother

\newtheorem{theorem}{Theorem}
\newtheorem{lemma}{Lemma}
\newtheorem{remark}{Remark}

\def\R{\mathbb{R}}

\newcommand{\wt}{\widetilde}

\usepackage{color}

\usepackage[dvips]{epsfig}
\usepackage{graphicx}

\begin{document}

\title{\bf Normal form for the symmetry-breaking bifurcation \\ in the nonlinear Schr\"{o}dinger equation}

\author{D.E. Pelinovsky$^1$ and T. Phan$^2$ \\
{\small $^{1}$ Department of Mathematics, McMaster University, Hamilton, Ontario, Canada, L8S 4K1}\\
{\small $^{2}$ Department of Mathematics, University of Tennessee, Knoxville, TN 37996} }

\maketitle

\begin{abstract}
We derive and justify a normal form reduction of the nonlinear Schr\"{o}dinger equation for
a general pitchfork bifurcation of the symmetric bound state that occurs in a double-well symmetric potential.
We prove persistence of normal form dynamics for both supercritical and subcritical pitchfork bifurcations
in the time-dependent solutions of the nonlinear Schr\"{o}dinger equation over long but finite time intervals.
\end{abstract}

\section{Introduction}

We consider the nonlinear Schr\"{o}dinger (NLS) equation with a focusing power
nonlinearity and an external potential (also known as the Gross-Pitaevskii equation),
\begin{equation}\label{GP}
i \Psi_t = -\Psi_{xx} + V(x) \Psi - |\Psi|^{2p} \Psi,
\end{equation}
where $\Psi(x,t) : \R \times \R \to \mathbb{C}$ is the wave
function, $p \in \mathbb{N}$ is the nonlinearity power,  and
$V(x) :\mathbb{R} \rightarrow \mathbb{R}$ is the external, symmetric, double-well potential
satisfying the following conditions:
\begin{enumerate}
\item[\textup{(H1)}] $V(x) \in L^{\infty}(\R)$ and $xV'(x) \in L^\infty(\mathbb{R})$;
\item[\textup{(H2)}] $\lim_{|x| \rightarrow \infty }V(x) =0$;
\item[\textup{(H3)}] $V(-x) = V(x)$ for all $x \in \R$;
\item[\textup{(H4)}] $L_0 = -\partial_x^2 + V(x)$ has the lowest eigenvalue $-E_0 <0$;
\item[\textup{(H5)}] $V(x)$ has a non-degenerate local maximum at $x = 0$ and two minima at $x = \pm
x_0$ for some $x_0 > 0$.
\end{enumerate}

The easiest way to think about
the double-well potential $V(x)$ is to represent it with
\begin{equation}\label{potential}
V(x) = \frac{1}{2} \left[ V_0(s-x) + V_0(s+x) \right], \quad s
\geq 0,
\end{equation}
where the single-well potential $V_0(x)$ satisfies (H1)--(H4) and has a global minimum at $x
= 0$ and no other extremum points. For sufficiently large $s > s_*$, where $s_*$ is the
inflection point of $V_0$, that is, $V_0''(s_*) = 0$, the sum of
two single-well potentials (\ref{potential}) becomes a
double-well potential we would like to consider.

The symmetric double-well potentials are used in the atomic physics of Bose--Einstein
condensation \cite{markus1} through a combination of parabolic and periodic (optical
lattice) potentials. Similar potentials were also examined in the
context of nonlinear optics, e.g. in optically induced photo-refractive crystals
\cite{zhigang} and in a structured annular core of an optical fiber \cite{Longhi}.
Physical relevance and simplicity of the model make the topic fascinating
for a mathematical research.

Bifurcations of stationary states and their stability in the NLS equation (\ref{GP})
under the assumptions (H1)--(H5) on the potential $V(x)$ were recently considered by
Kirr {\em et al.} \cite{KirrPelin}.

Let $\Psi(x,t) = e^{i E t} \phi(x;E)$ be a stationary state such that
$\phi(x;E)$ is a solution of the stationary nonlinear Schr\"{o}dinger equation
\begin{equation}\label{stationary}
(-\partial_x^2 + V) \phi - \phi^{2p+1}  + E \phi = 0.
\end{equation}
Via standard regularity theory, if $V \in L^{\infty}(\R)$, then
any weak solution $\phi(\cdot;E) \in H^1(\R)$ of the stationary equation
\eqref{stationary} belongs to $H^2(\mathbb R)$. Moreover, if $-E
\notin \sigma(L_0)$, then the solution $\phi(\cdot;E) \in H^2(\R)$ decays
exponentially fast to zero as $|x| \to \infty$.

Existence of symmetric stationary states $\phi$ for any $E > E_0$
bifurcating from the lowest eigenvalue $-E_0$ of the operator $L_0 = -\partial_x^2 + V(x)$
was first considered by Jeanjean and Stuart \cite{js:ubs}.
Kirr {\em et al.} \cite{KirrPelin} continued this research theme and obtained the
following bifurcation theorem.

\begin{theorem}[{\em Kirr et all. \cite{KirrPelin}}]
Consider the stationary NLS equation (\ref{stationary}) with $p \geq \frac{1}{2}$
and $V(x)$ satisfying (H1)--(H5).

(i) There exists a $C^1$ curve
$(E_0,\infty) \ni E \mapsto \phi(\cdot;E) \in H^2(\R)$ of positive symmetric
states bifurcating from the zero solution at $E = E_0$.
This curve undertakes the symmetry--breaking (pitchfork) bifurcation at
a finite $E_* \in (E_0,\infty)$, for which the second eigenvalue of the
operator
\begin{equation}\label{Jacobian}
L_+(E) = -\partial_x^2 + V(x) - (2p+1) \phi^{2p}(x;E) + E
\end{equation}
passes from positive values for $E < E_*$ to negative values for
$E > E_*$.

(ii) Let $\phi_*(x) = \phi(x;E_*)$ be the positive symmetric state at the
bifurcation point and $\psi_* \in H^2(\R)$ be the anti-symmetric eigenvector
of $L_+(E_*)$ corresponding to the second eigenvalue
$\lambda(E)$ such that $\lambda(E_*) = 0$ and $\lambda'(E_*) < 0$.
The $C^1$ curve $(E_0,\infty) \ni E \mapsto \phi(\cdot;E) \in H^2(\R)$
intersects transversely at $E = E_*$ with the $C^1$ curve of
positive asymmetric states $E \mapsto \varphi_{\pm}(\cdot;E) \in H^2(\R)$
that extends to $E > E_*$ if ${\cal Q} < 0$ and
to $E < E_*$ if ${\cal Q} > 0$, where
\begin{equation}
\label{bifurcation-coefficient}
{\cal Q} = 2 p^2 (2p+1)^2 \langle \phi_*^{2p-1} \psi_*^2, L_+^{-1}(E_*)
\phi_*^{2p-1} \psi_*^2 \rangle_{L^2} + \frac{1}{3} p (2p+1) (2p-1)
\langle \psi_*^2, \phi_*^{2p-2} \psi_*^2 \rangle_{L^2}.
\end{equation}
The asymmetric states $\varphi_+$ and $\varphi_-$ are centered at the left and
the right well of $V$, respectively. \label{theorem-Kirr}
\end{theorem}

Orbital stability of the stationary state $\phi(x;E)$
in the NLS equation (\ref{GP}) depends on the number of negative
eigenvalues of $L_+(E)$ and $L_-(E)$, where
\begin{equation}\label{Operator}
L_-(E) = -\partial_x^2 + V(x) - \phi^{2p}(x;E) + E.
\end{equation}
Since $L_-(E) \phi(E) = 0$ and $\phi(x;E) > 0$ for all $x \in \R$ and $E > E_0$,
the spectrum of $L_-(E)$ is non-negative for any $E > E_0$.
This fact simplifies the stability analysis of the stationary states.

Let us denote $N_s(E) = \| \phi(\cdot;E) \|_{L^2}^2$
and $N_a(E) = \| \varphi_+(\cdot;E) \|_{L^2}^2 = \| \varphi_-(\cdot;E) \|_{L^2}^2$.
In what follows, we always assume that
$$
N_s'(E_*) = 2 \langle \partial_E \phi_*, \phi_* \rangle_{L^2} > 0, \quad \mbox{\rm where} \quad
\partial_E \phi_*(x) = \partial_E \phi(x;E_*),
$$
that is, $N_s(E)$ is increasing near the bifurcation point $E = E_*$.
For example, this assumption is satisfied for large separation distance
$s$ between the two wells of $V$ given by (\ref{potential}), because $E_* \to E_0$
as $s \to \infty$ and $N_s(E_0) = 0$. The following stability theorem was also proven in \cite{KirrPelin}.

\begin{theorem}[{\em Kirr et all. \cite{KirrPelin}}]
Assume $N_s'(E_*) > 0$ in addition to conditions of Theorem \ref{theorem-Kirr}.
Then the symmetric state $\phi$ is orbitally stable for $E \leq E_*$ and unstable
for $E > E_*$. If in addition, ${\cal Q} < 0$, then $N_a(E)$ is an increasing function
of $E > E_*$ if ${\cal S} > 0$
and it is a decreasing function of $E > E_*$ if ${\cal S} < 0$, where
\begin{equation}
\label{represent-0}
{\cal S} = N_s'(E_*) + {\cal Q}^{-1} \left( \lambda'(E_*) \| \psi_* \|^2_{L^2} \right)^2.
\end{equation}
Consequently, the asymmetric states $\varphi_{\pm}$ near $E = E_*$
are orbitally stable for ${\cal S} > 0$ and unstable for ${\cal S} < 0$. \label{theorem-stability}
\end{theorem}

For any potential $V(x)$ represented by (\ref{potential}) with sufficiently large $s \to \infty$,
it was found in \cite{KirrPelin} that ${\cal Q} < 0$ for any $p \geq \frac{1}{2}$, hence,
the stable symmetric state $\phi$ for $E < E_*$ becomes unstable for $E > E_*$ and
the asymmetric states $\varphi_{\pm}$ exist for $E > E_*$. In the limit $s \to \infty$,
the boundary ${\cal S} = 0$ is equivalent to $p = p_*$, where
\begin{equation}\label{threshold}
p_* = \frac{3 + \sqrt{13}}{2} \approx 3.3028.
\end{equation}
If $p < p_*$, the asymmetric states $\varphi_{\pm}$ are stable for $E > E_*$.
If $p > p_*$, both symmetric and asymmetric states
are unstable for $E > E_*$. Therefore, we can classify the symmetry-breaking
bifurcation at $E = E_*$ as the supercritical (if ${\cal S} > 0$) or the subcritical
(if ${\cal S} < 0$) pitchfork bifurcations with respect to the squared $L^2$-norm,
which is a conserved quantity of the NLS equation (\ref{GP}) in time.
The functions $N_s(E)$ and $N_a(E)$ in the two different cases are shown
schematically on Figure \ref{figure-1}, where stable
branches are depicted by solid line and the unstable branches are depicted by dotted lines.

\begin{figure}
\begin{center}
\includegraphics[height=8cm]{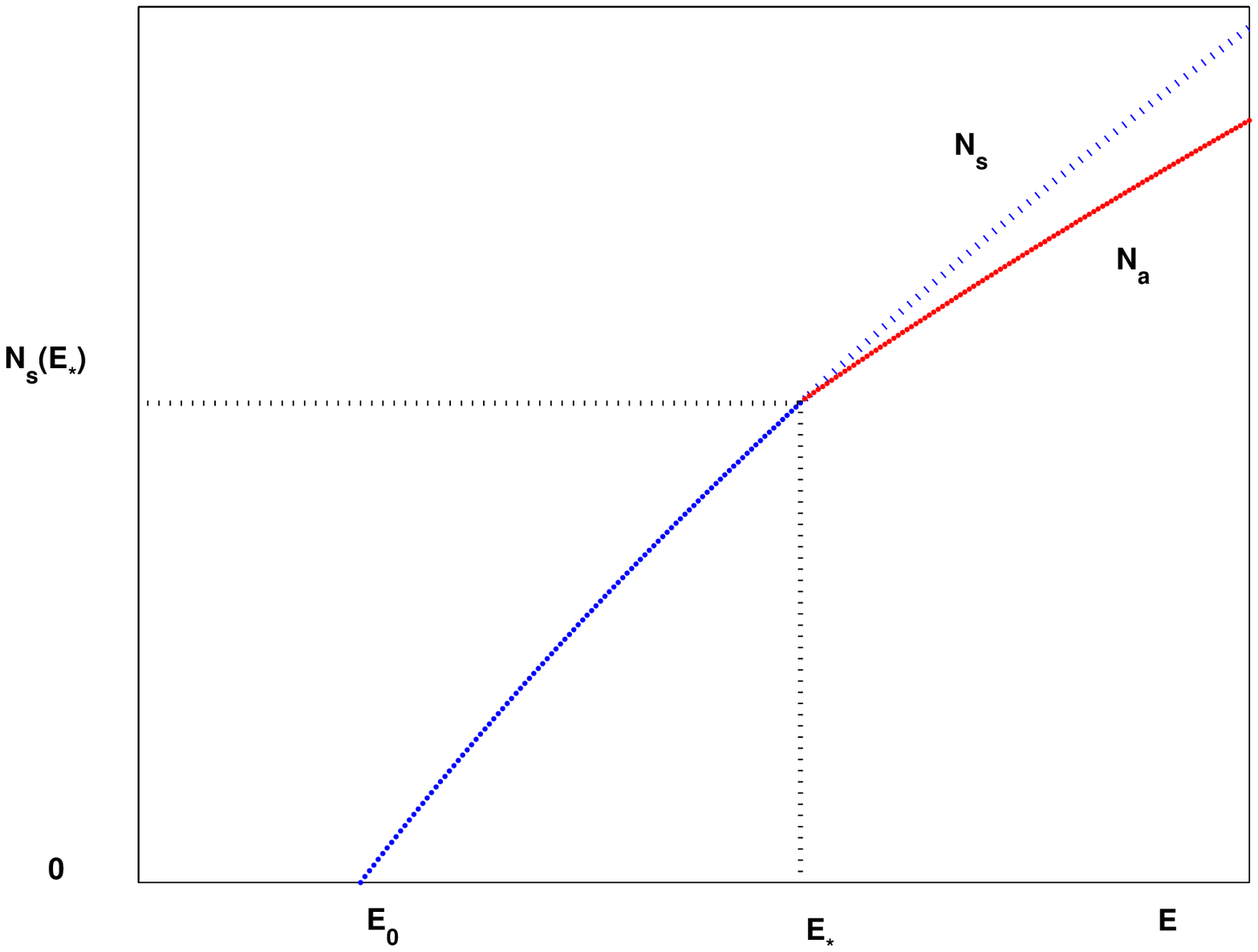}
\includegraphics[height=8cm]{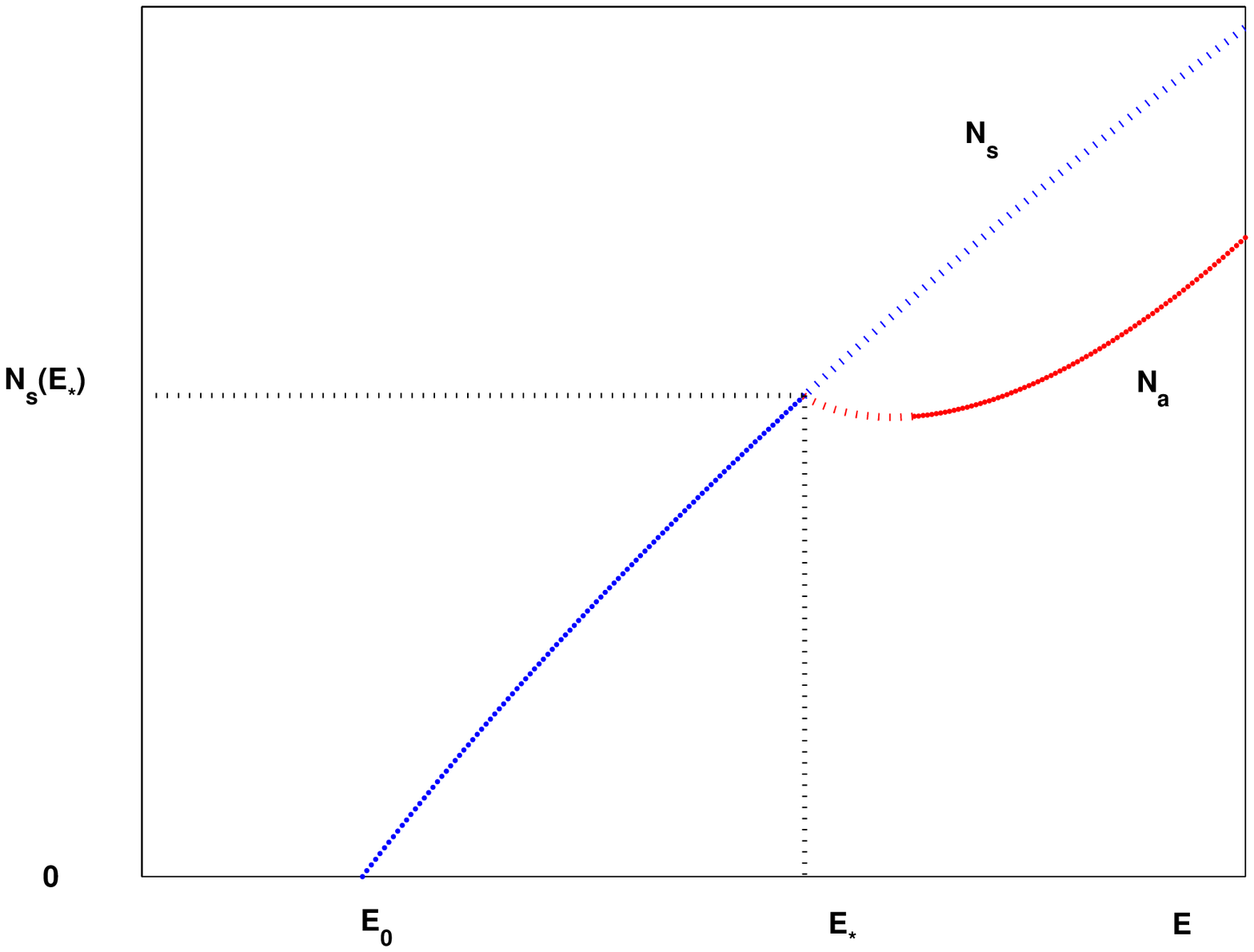}
\end{center}
\caption{Schematic representation of the supercritical (top)
and subcritical (bottom) pitchfork bifurcations in terms of the functions $N_s(E)$ and $N_a(E)$.
Unstable stationary states are shown by dotted curves.}
\label{figure-1}
\end{figure}

The classification into the supercritical
and subcritical pitchfork bifurcations is usually based on the analysis of the
normal form equations obtained from the center
manifold reductions and the near identity transformations. This analysis is the goal of this paper.
We shall look at the long but finite temporal dynamics of the normal form equations,
avoiding the complexity of the time evolution at infinite time intervals.
To enable near identity transformations up to any polynomial order, we shall
only consider the integer values of $p$.

The normal form equations have been considered previously in a
similar context. In the limit of large separation of the two potential wells, Kirr {\em
et al.} \cite{Kirr} derived a two-mode reduction of the
NLS equation. Persistence of this reduction
for periodic small-amplitude oscillations near stable
stationary states was addressed by Marzuola \& Weinstein
\cite{MW10}. Similar but more formal reduction to the two-mode
equations was developed by Sacchetti \cite{Sacchetti} using the
semi-classical analysis. In comparison with \cite{Kirr,MW10}, Sacchetti \cite{Sacchetti}
considered the defocusing version of the NLS equation,
where the anti-symmetric stationary state undertakes a similar
symmetry-breaking bifurcation. Based on the two-mode reduction,
Sacchetti \cite{Sacchetti2} also reported the same threshold $p_*$ as in (\ref{threshold})
that separates the supercritical and subcritical pitchfork
bifurcations.

Unlike these previous works, we shall deal with a {\em general}
symmetry-breaking bifurcation of the symmetric states. We develop
a simple but robust analysis, which justifies a general normal form equation
for the pitchfork bifurcation. {\em Arbitrary} bounded
solutions of the normal form equation are proved to shadow dynamics
of time-dependent solutions of the NLS equation (\ref{GP})
near the stationary bound states for long but finite time intervals.
Previously, only small-amplitude periodic solutions of the normal form equation were considered
within the two-mode approximations in the large separation limit \cite{MW10}.
Also, compared with the sophisticated analysis based on Strichartz estimates
and wave operators for the linear Schr\"{o}dinger equations in \cite{MW10}, our analysis
is only based on the spectral decompositions and Gronwall inequalities.
Thus, we show how basic analytical methods can be used to treat time-dependent
normal form equations for bifurcations in the nonlinear Schr\"{o}dinger equations.
Our main result is formulated in the following theorem, where we use notations of
Theorems \ref{theorem-Kirr} and \ref{theorem-stability}.

\begin{theorem}
\label{theorem-main} Assume that there exists $E_* \in (E_0,\infty)$ such that
$\lambda'(E_*) < 0$, $N_s'(E_*) > 0$, and ${\cal Q} < 0$. Fix ${\cal N}_0$ and define
$\Delta N = {\cal N}_0 - N_s(E_*)$. There exists $\varepsilon > 0$ such that for any
$|\Delta N| < \varepsilon$, there exists $T > 0$, $\Psi_0 \in H^1$ with ${\cal N}_0 = \| \Psi_0 \|_{L^2}^2$,
and functions $(\theta,E,A,B) \in C^1([0,T];\R^4)$ such that the NLS equation (\ref{GP}) admits
a solution $\Psi \in C([0,T];H^1(\R))$ with $\Psi(x,0) = \Psi_0(x)$ in the form
$$
\Psi(x,t) = e^{i \theta(t)} \left[ \phi(x;E(t)) + A \psi(x;E(t)) + i B \chi(x;E(t)) \right]
+ \tilde{\Psi}(x,t),
$$
where $\psi(x;E)$ and $\chi(x;E)$ satisfy
$$
L_+(E) \psi = - \Lambda^2(E) \chi, \quad L_-(E) \chi = \psi,
$$
subject to the normalization $\langle \chi,\psi\rangle_{L^2} = 1$, with the asymptotic expansion,
$$
\Lambda^2(E) = -\lambda'(E_*) \|\psi_*\|^2_{L^2} (E - E_*) + {\cal O}(E- E_*)^2 \quad \mbox{\rm as} \quad E \to E_*.
$$
Moreover, there are positive constants $C_0$, $C_1$, $C_2$, $C_3$, and $C_4$
such that $T \leq C_0 |\Delta N|^{-1/2}$,
$$
\| \tilde{\Psi}(\cdot,t) \|_{H^1} \leq C_1 |\Delta N|, \quad |\dot{\theta}(t) - E_*| \leq C_2 |\Delta N|, \quad
|E(t) - E_*| \leq C_3 |\Delta N|, \quad \mbox{\rm for all} \;\; t \in [0,T],
$$
and the trajectories of $(A,B)$ in the ellipsoidal domain,
$$
D = \left\{ (A,B) \in \R^2 : \quad A^2 + |\Delta N|^{-1} B^2 \leq C_4 |\Delta N| \right\},
$$
are homeomorphic to those of the second-order system,
\begin{equation}
\label{normal-form-equation-time}
\left\{ \begin{array}{l} \dot{A} = B, \\
\dot{B} = \left(-\lambda'(E_*) \|\psi_*\|^2_{L^2} (\Delta N) A + {\cal Q} {\cal S} A^3\right)/N_s'(E_*). \end{array} \right.
\end{equation}
\end{theorem}

Trajectories of the second-order system (\ref{normal-form-equation-time}) with $\lambda'(E_*) < 0$
and $N'_s(E_*) > 0$ are shown on Figure \ref{figure-2} for four distinct cases of different values of
$\Delta N$ and ${\cal Q} {\cal S}$.

The article is organized as follows.
In Section 2, we use Theorem \ref{theorem-Kirr}(i) and
derive modulation equations for dynamics of time-dependent solutions of the
NLS equation near the stationary bound states at the onset of
the symmetry-breaking (pitchfork) bifurcation.
In Section 3, we consider the stationary modulation equations and recover the results of
Theorems \ref{theorem-Kirr}(ii) and \ref{theorem-stability}
from our system of equations. In Section 4, we justify the dynamics of the time-dependent modulation equations
and give a proof of Theorem \ref{theorem-main}.

\begin{figure}
\begin{center}
\includegraphics[height=6.5cm]{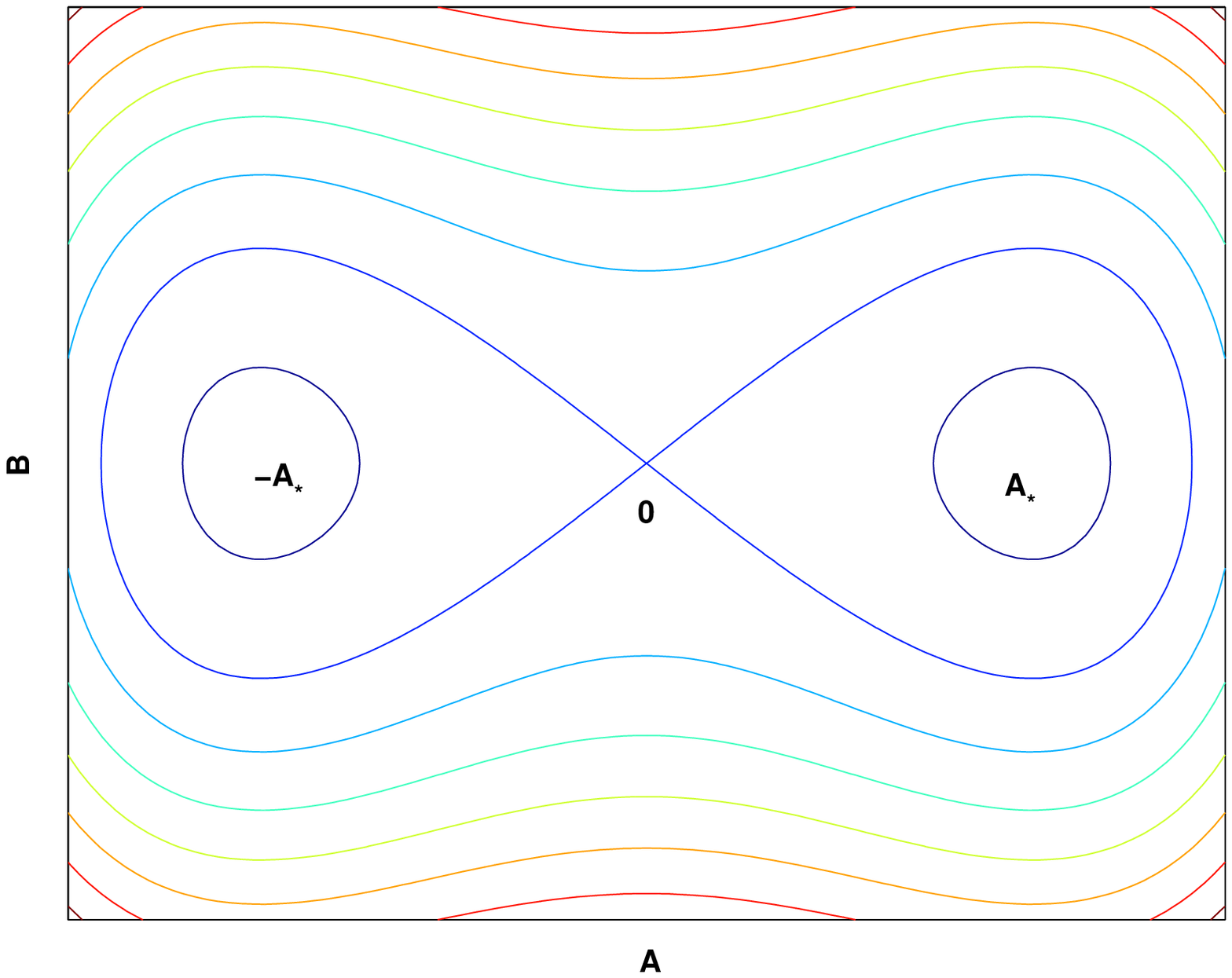} $\;\;$
\includegraphics[height=6.5cm]{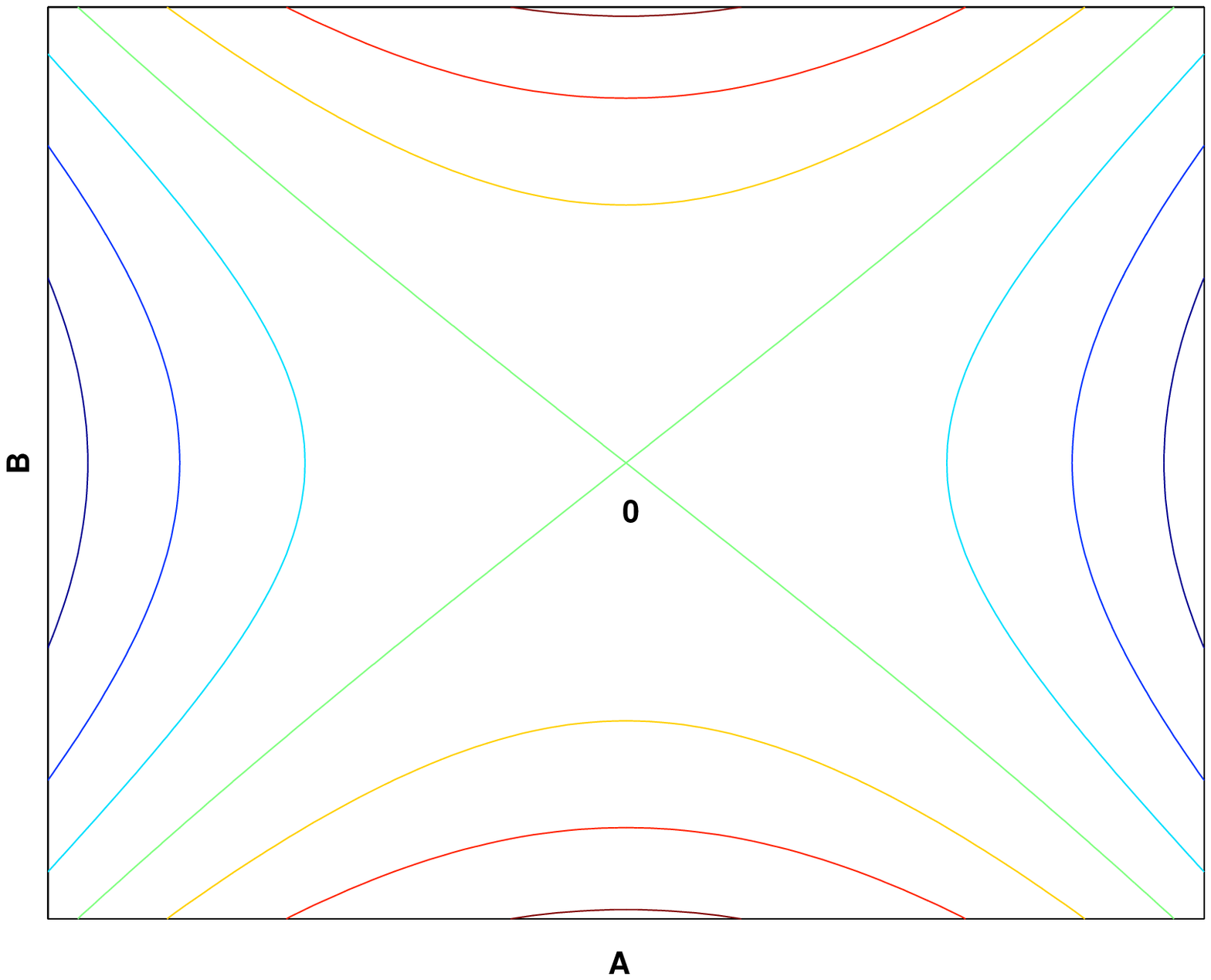}\\ \vspace{1cm}
\includegraphics[height=6.5cm]{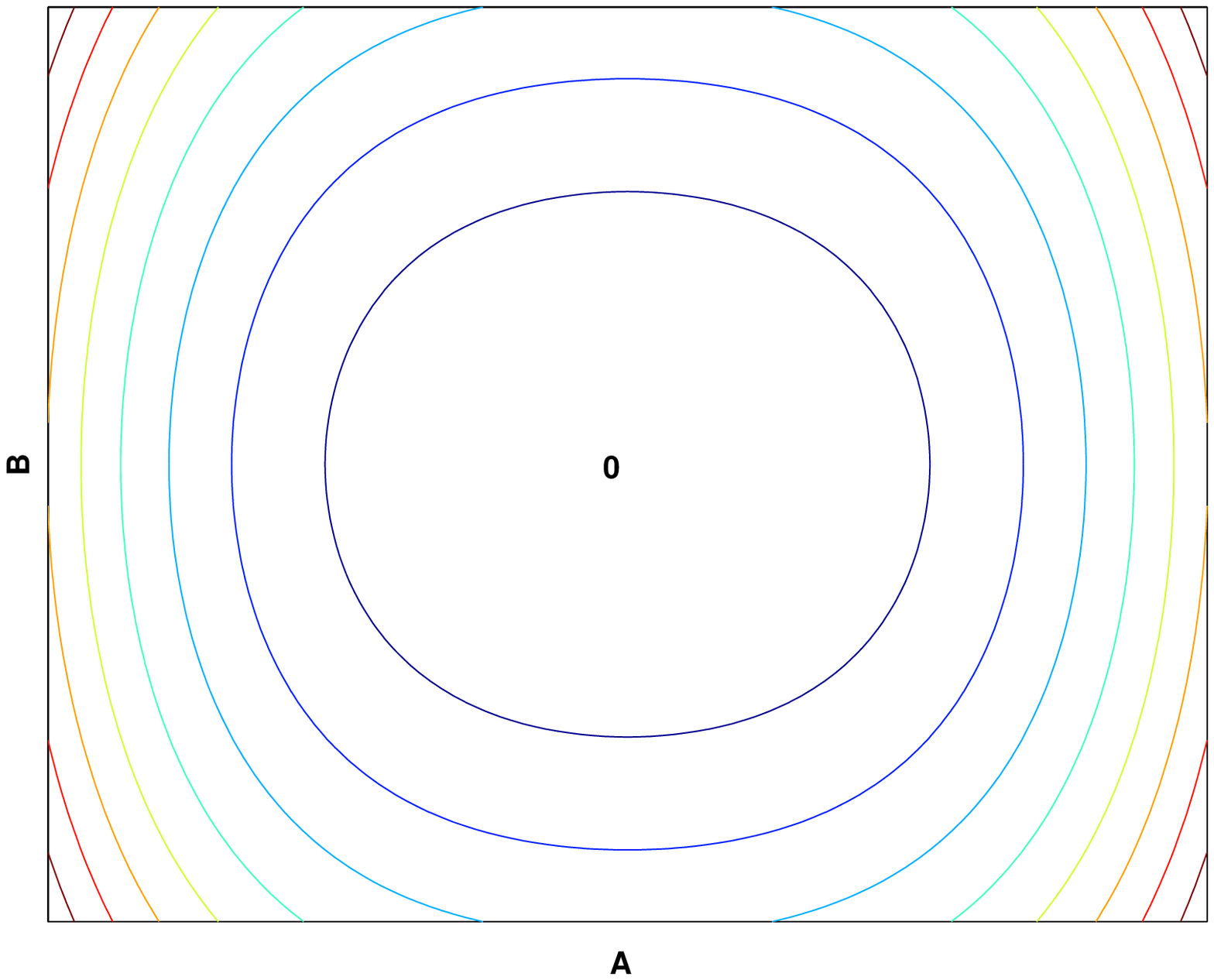} $\;\;$
\includegraphics[height=6.5cm]{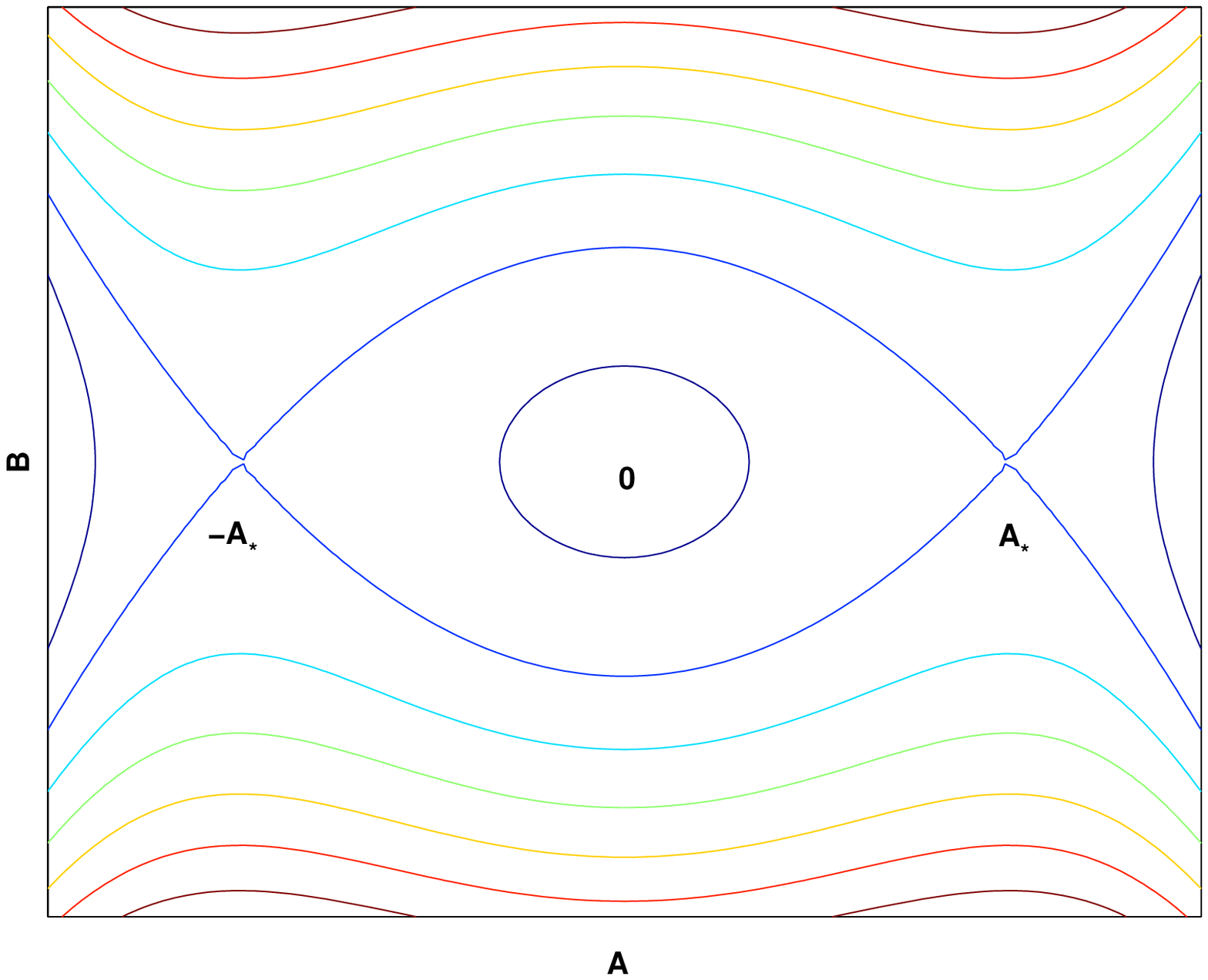}
\end{center}
\caption{Trajectories of the second-order system (\ref{normal-form-equation-time})
on the phase plane $(A,B)$ for $\Delta N > 0$ and ${\cal Q} {\cal S} < 0$ (top left);
$\Delta N > 0$ and ${\cal Q} {\cal S} > 0$ (top right);
$\Delta N < 0$ and ${\cal Q} {\cal S} < 0$ (bottom left);
$\Delta N < 0$ and ${\cal Q} {\cal S} < 0$ (bottom right).}
\label{figure-2}
\end{figure}

\section{Modulation equations for dynamics of bound states}

We shall derive a set of modulation equations which describe temporal dynamics of solutions of
the NLS equation (\ref{GP}) near the stationary bound states at the onset
of the symmetry-breaking bifurcation. We only use the statement of Theorem \ref{theorem-Kirr}(i)
on the existence of the symmetry-breaking bifurcation for the symmetric stationary state
of the NLS equation (\ref{GP}) under assumptions (H1)--(H5) and $p \in \mathbb{N}$.

\subsection{Primary decomposition near the symmetric stationary state}

Let $\phi(x;E)$ be a solution of the stationary NLS equation
(\ref{stationary}) with properties
$$
\phi(\cdot;E) \in H^2(\mathbb{R}) : \quad \phi(-x;E) = \phi(x;E) > 0 \quad \mbox{\rm for all} \;\; x \in \R.
$$
It is stated in Theorem \ref{theorem-Kirr}(i) that the $C^1$ curve
$E \mapsto \phi(\cdot;E) \in H^2(\R)$ exist for all $E \in (E_0,\infty)$.
If $p \in \mathbb{N}$, this curve is actually $C^{\infty}$ by the bootstrapping arguments.

We shall consider a solution of the NLS equation (\ref{GP}) in the form
\begin{equation}
\label{decomposition1}
\Psi(x,t) = e^{i \theta(t)} \left[ \phi(x;E(t)) + u(x,t) + i w(x,t) \right],
\end{equation}
where $(E,\theta)$ are coordinates of the stationary state. Direct substitution
of (\ref{decomposition1}) into (\ref{GP}) shows that the real functions
$(u,w)$ satisfy the system of time evolution equations
\begin{eqnarray}
\label{evolution-u}
u_t & = & L_- w + N_-(u,w) + (\dot{\theta} - E) w - \dot{E} \partial_E \phi, \\\label{evolution-w}
-w_t & = & L_+ u + N_+(u,w) + (\dot{\theta} - E) (\phi + u),
\end{eqnarray}
where $L_+$ and $L_-$ are defined by (\ref{Jacobian}) and (\ref{Operator}) and
the nonlinear terms are given explicitly by
\begin{eqnarray*}
N_+(u,w) & = & -(\phi + u)( \phi^2 + 2 \phi u + u^2 + w^2)^p + \phi^{2p} (\phi + (2p+1) u), \\
N_-(u,w) & = & -w [( \phi^2 + 2 \phi u + u^2 + w^2)^p - \phi^{2p}].
\end{eqnarray*}
For any $p \in \mathbb{N}$, we can use the Taylor series expansions
\begin{eqnarray}
\nonumber
N_+(u,w) & = & -p(2p+1) \phi^{2p-1} u^2 - p \phi^{2p-1} w^2 \\
\label{power-series-u} & \phantom{t} &
-\frac{1}{3} p (2p+1) (2p-1) \phi^{2p-2} u^3 - p(2p-1) \phi^{2p-2} u w^2 + {\cal O}(u^2+w^2)^2, \\
\label{power-series-w}
N_-(u,w) & = & -2p \phi^{2p-1} u w - p(2p-1) \phi^{2p-2} u^2 w - p \phi^{2p-2} w^3 + {\cal O}(u^2+w^2)^2.
\end{eqnarray}

To determine $(E,\theta)$ uniquely in the neighborhood of the stationary state
(for small $u$ and $w$), we add the standard conditions of symplectic orthogonality
\begin{eqnarray}
\label{orthogonality1}
\langle \phi, u \rangle_{L^2} = 0, \quad \langle \partial_E \phi, w \rangle_{L^2} = 0,
\end{eqnarray}
where we recall that
\begin{equation}
\label{null-space}
L_- \phi = 0, \quad L_+ \partial_E \phi = - \phi.
\end{equation}
Under symplectic orthogonality conditions (\ref{orthogonality1}),
the rate of changes of $(E,\theta)$ are uniquely determined from
the projection equations
\begin{eqnarray}
\label{projection1}
\left[ \begin{array}{cc} \langle \partial_E \phi, \phi - u \rangle_{L^2} &  -\langle \phi, w \rangle_{L^2} \\
-\langle \partial_E^2 \phi, w \rangle_{L^2} & \langle \partial_E \phi, \phi + u \rangle_{L^2} \end{array} \right]
\left[ \begin{array}{c} \dot{E} \\ \dot{\theta} - E \end{array} \right] =
\left[ \begin{array}{cc} \langle \phi, N_-(u,w) \rangle_{L^2} \\
- \langle \partial_E \phi, N_+(u,w) \rangle_{L^2} \end{array} \right].
\end{eqnarray}

We shall now study eigenvectors at the onset of the symmetry-breaking bifurcation
in order to build a frame for the secondary decomposition of the perturbations
$(u,w)$ near these eigenvectors.

\subsection{Linear eigenvectors}

It is stated in Theorem \ref{theorem-Kirr}(i) that there exists
a bifurcation value $E_* \in (E_0,\infty)$ such that the second eigenvalue
$\lambda(E)$ of $L_+(E)$ satisfies $\lambda(E_*) = 0$.
We shall denote $\phi_*(x) =\phi(x;E_*)$ at the bifurcation value $E = E_*$.
In many cases, we will suppress the $x$-argument in the function $\phi(x;E)$
to underline the $E$-dependence of this function.
In this setting, we have the following result.

\begin{lemma} \label{psi*chi*}
There exist odd functions $\psi_*, \chi_* \in H^2(\mathbb{R})$ such that
\begin{equation}
L_+(E_*)\psi_* =0, \quad L_-(E_*) \chi_* =\psi_*,
\quad \langle \chi_*, \psi_* \rangle_{L^2} =1.  \label{EF-EV-limiting}
\end{equation}
Moreover,
\begin{equation}
\label{derivative-result}
\lambda'(E_*) = 1 - 2p (2p+1) \frac{\langle \partial_E \phi_*, \phi_*^{2p-1} \psi_*^2 \rangle_{L^2}}{\| \psi_* \|^2_{L^2}}.
\end{equation}
\end{lemma}

\begin{proof}
Let $g(x;E)$ be an eigenfunction of $L_+(E)$ with eigenvalue $\lambda(E)$.
By Sturm's Theorem, $g(x;E)$ is odd in $x$ because $\lambda(E)$ is the second eigenvalue
of $L_+(E)$. Let $\psi_*(x) = g(x;E_*)$. So, $\psi_*$ is an odd function such that
$L_+(E_*)\psi_* = \lambda(E_*)\psi_* =0$.

Since $\phi_*(x) =\phi(x;E_*)$ is even, positive and $L_-(E_*)\phi_* =0$,
we see that zero is the lowest eigenvalue of $L_-(E_*)$.
Therefore, it is a simple eigenvalue. Because $\psi_*$ is odd, there is an odd function $\chi_* \in H^2(\mathbb{R})$
such that $L_-(E_*) \chi_* =\psi_*$. On the other hand, we have
$$
\langle \chi_*, \psi_* \rangle_{L^2} = \langle L_-^{-1}(E_*) \psi_*, \psi_* \rangle_{L^2} > 0.
$$
By rescaling $\psi_*$ and $\chi_*$, we get $\langle \chi_*, \psi_* \rangle_{L^2} =1$. This completes the proof of the first part of the lemma.

To prove \eqref{derivative-result}, we note that due to the smooth continuation of $\phi(x;E)$ across $E_*$,
we can compute explicitly,
\begin{equation} \label{exL}
L'_+(E_*) = 1 - 2p (2p+1) \phi_*^{2p-1} \partial_E \phi_*, \quad
L'_-(E_*) = 1 - 2p \phi_*^{2p-1} \partial_E \phi_*.
\end{equation}
By differentiating the relation $L_+(E)g(E) = \lambda(E) g(E)$ at $E = E_*$, we get
\[
L_+'(E_*) \psi_* + L_+(E_*) \partial_E\psi_* = \lambda'(E_*) \psi_*.
\]
Taking the inner product of this equation with $\psi_*$, we get
\[
\lambda'(E_*) = \frac{\langle L_+'(E_*) \psi_*, \psi_* \rangle_{L^2}}{\| \psi_* \|^2_{L^2}}
= 1 - 2p (2p+1) \frac{\langle \partial_E \phi_*, \phi_*^{2p-1} \psi_*^2 \rangle_{L^2}}{\| \psi_* \|^2_{L^2}}.
\]
This completes the proof of the lemma.
\end{proof}

We would like now to extend the functions $(\psi_*,\chi_*)$ as the
eigenvectors of the linearized system associated with the
time-dependent system (\ref{evolution-u}) and (\ref{evolution-w}) near $E = E_*$.
Note that the eigenvectors of the linearized (non-self-adjoint) system are
different from the eigenvector $g$ of the (self-adjoint) operator $L_+(E)$ introduced in the
proof of Lemma \ref{psi*chi*}. The following lemma gives the extension of $(\psi_*,\chi_*)$ near $E = E_*$.

\begin{lemma} \label{linearized-functions}
There exists sufficiently small $\epsilon > 0$ such that for all $|E-E_*|< \epsilon$,
there exists a small eigenvalue $\Lambda(E)$ of the linearized system
\begin{eqnarray}
\label{linearized-system}
L_+(E) \psi(E) = -\Lambda^2(E) \chi(E), \quad L_-(E) \chi(E) = \psi(E),
\end{eqnarray}
where the eigenvector and eigenvalue satisfy the asymptotic expansion,
\begin{equation} \label{psichi-ex}
\psi(E) = \psi_* + {\cal O}_{H^2}(E-E_*), \quad \chi(E) = \chi_* + {\cal O}_{H^2}(E-E_*), \quad
\langle \chi(E), \psi(E) \rangle_{L^2} = 1,
\end{equation}
and
\begin{equation}
\label{perturbation-theory}
\Lambda^2(E) = -\lambda'(E_*) \| \psi_* \|^2_{L^2}(E-E_*) + {\cal O}(E - E_*)^2.
\end{equation}
Consequently, if $\lambda'(E_*) < 0$, the eigenvalue $\Lambda(E)$ is real for $E > E_*$
and purely imaginary for $E < E_*$.
\end{lemma}

\begin{proof}
Recall that $L_-(E_*) \phi_* = 0$, $L_+(E_*) \psi_* = 0$ and $L_-(E_*) \chi_* = \psi_*$, where
both $\psi_*$ and $\chi_*$ are odd in $x$ and $\phi_*$ is even in $x$. By the perturbation theory
for isolated eigenvalues, there exists a solution for $\Lambda(E)$, $\psi(E)$, and $\chi(E)$
in the system (\ref{linearized-system}) such that the eigenvectors satisfy the expansion \eqref{psichi-ex}.

By taking the inner product of $L_+(E)\psi(E) = -\Lambda^2(E)\chi(E)$ with $\psi(E)$ and using
\eqref{derivative-result}, \eqref{exL}, and \eqref{psichi-ex}, we get (\ref{perturbation-theory}).
So, it follows that if $\lambda'(E_*) < 0$ the eigenvalue $\Lambda(E)$ is real for $E > E_*$ and purely
imaginary for $E < E_*$.

On the other hand, for small $|E-E_*|$, we have
$$
\langle \chi(E), \psi(E) \rangle_{L^2} = \langle \chi_*, \psi_* \rangle_{L^2} + {\cal O}(E-E_*)
= 1 + {\cal O}(E-E_*) > 0.
$$
We can hence normalize $\psi(E)$ and $\chi(E)$ such that
$\langle \chi(E), \psi(E) \rangle_{L^2} = 1$. This completes the proof of the lemma.
\end{proof}

\begin{remark}
Under the normalization $\langle \chi(E), \psi(E) \rangle_{L^2} = 1$,
the $L^2$ norms of $\psi$ and $\chi$ are no longer normalized to unity,
in comparison with the normalization used in \cite{KirrPelin}.
\end{remark}

\subsection{Secondary decomposition near the linear eigenvectors}

Let us now decompose the perturbation terms into
\begin{equation}
\label{decomposition2}
u(x,t) = A(t) \psi(x;E) + U(x,t), \quad w(x,t) = B(t) \chi(x;E) + W(x,t),
\end{equation}
where $(A,B)$ are coordinates of the decomposition and
$(U,W)$ are the remainder terms. The linear eigenvectors
$(\psi,\chi)$ are solutions of the linearized system (\ref{linearized-system})
for $E$ near $E_*$. The remainder terms $(U,W)$
are required to satisfy the conditions of symplectic orthogonality
\begin{eqnarray}
\label{orthogonality2}
\langle \phi, U \rangle_{L^2} = 0, \quad \langle \partial_E \phi, W \rangle_{L^2} = 0, \quad
\langle \chi, U \rangle_{L^2} = 0, \quad \langle \psi, W \rangle_{L^2} = 0.
\end{eqnarray}

Substitution of (\ref{decomposition2}) into (\ref{evolution-u})--(\ref{evolution-w}) show
that $(U,W)$ satisfy the time-evolution equations,
\begin{eqnarray}
\nonumber
U_t & = & L_- W + N_-(A \psi + U,B \chi + W) \\
\label{time-W-U} & \phantom{t} & \phantom{text} + (\dot{\theta} - E) (B \chi + W)
- \dot{E} (\partial_E \phi + A \partial_E \psi) - (\dot{A} - B) \psi, \\
\nonumber
-W_t & = & L_+ U + N_+(A \psi + U,B \chi + W) \\
\label{time-U-W} & \phantom{t} & \phantom{text}  + (\dot{\theta} - E) (\phi + A \psi + U)
+ B \dot{E} \partial_E \chi + (\dot{B} - \Lambda^2 A) \chi.
\end{eqnarray}

Under the orthogonality conditions (\ref{orthogonality2}),
the rate of changes of $(E,\theta,A,B)$ are uniquely determined from
the projection equations
\begin{eqnarray} \label{system-theta-E}
\label{projection2}
{\cal M} \left[ \begin{array}{c} \dot{E} \\ \dot{\theta} - E \end{array} \right] =
\left[ \begin{array}{cc} \langle \phi, N_-(A \psi + U,B \chi + W) \rangle_{L^2} \\
- \langle \partial_E \phi, N_+(A \psi + U,B \chi + W) \rangle_{L^2} \end{array} \right],
\end{eqnarray}
where
\begin{equation} \label{A.def}
 \mathcal{M} = \left[ \begin{array}{cc} \langle \partial_E \phi, \phi - U \rangle_{L^2} &  -\langle \phi, W \rangle_{L^2} \\ -\langle \partial_E^2 \phi, W \rangle_{L^2} & \langle \partial_E \phi, \phi + U \rangle_{L^2} \end{array} \right],
\end{equation}
and
\begin{eqnarray}
\nonumber
\dot{A} - B & = & \langle \chi, N_-(A \psi + U,B \chi + W) \rangle_{L^2} \\  & \phantom{t} &
+ \dot{E} ( \langle \partial_E \chi, U \rangle_{L^2} - A \langle \partial_E \psi, \chi \rangle_{L^2}) +
(\dot{\theta} - E)( B \| \chi \|_{L^2}^2 + \langle \chi, W \rangle_{L^2}), \label{system-b-a} \\
\nonumber
\dot{B} - \Lambda^2 A & = & -\langle \psi, N_+(A \psi + U,B \chi + W) \rangle_{L^2} \\ & \phantom{t} &
+ \dot{E} ( \langle \partial_E \psi, W \rangle_{L^2} - B \langle \partial_E \chi, \psi \rangle_{L^2})
- (\dot{\theta} - E)( A \| \psi \|_{L^2}^2 + \langle \psi, U \rangle_{L^2}). \label{system-a-b}
\end{eqnarray}

In Section 4, we shall control the dynamics of small $(U,W)$,
$(E - E_*,\dot{\theta} - E_*)$, and $(A,B)$ in the system (\ref{time-W-U})--(\ref{system-a-b})
on long but finite time intervals.

\subsection{Conserved quantities}

The NLS equation (\ref{GP}) admits two conserved quantities given by
\begin{eqnarray}
\label{conservation-N}
{\cal N}[\Psi] = \int_{\R} |\Psi(x,t)|^2 dx
\end{eqnarray}
and
\begin{eqnarray}
\label{conservation-H}
{\cal H}[\Psi] = \int_{\R} \left[ |\Psi_x(x,t)|^2 +
V(x) |\Psi(x,t)|^2 - \frac{1}{p+1} |\Psi(x,t)|^{2p+2} \right] dx.
\end{eqnarray}
They are referred to as the power $N$ and the Hamiltonian $H$, respectively.

Let $N_s(E) = {\cal N}[\phi(\cdot;E)]$ and $H_s(E) = {\cal H}[\phi(\cdot;E)]$.
If $\Psi_0$ is an initial datum for the solution $\Psi$ of the NLS equation (\ref{GP}), we define
$$
{\cal N}_0 = {\cal N}[\Psi_0] \quad \mbox{\rm and} \quad {\cal H}_0 = {\cal H}[\Psi_0].
$$
Substitution of (\ref{decomposition1}) and (\ref{decomposition2})
into (\ref{conservation-N}) and (\ref{conservation-H}) gives
\begin{eqnarray*}
\label{conservation-N-expansion}
{\cal N}_0 = N_s(E) +
\int_{\R} \left[ ( A \psi + U)^2 + (B \chi + W)^2 \right] dx
\end{eqnarray*}
and
\begin{eqnarray*}
 \nonumber
{\cal H}_0 & = & H_s(E) + \int_{\R} \left[ (A \psi_x + U_x)^2 + (B \chi_x + W_x)^2
+ V  ( A \psi + U)^2 + V  (B \chi + W)^2 \right] dx \\
& \phantom{t} &  - \frac{1}{p+1} \int_{\R} \left[
\left( ( \phi + A \psi + U)^2 + (B \chi + W)^2 \right)^{p+1} - \phi^{2p+2} - 2(p+1)
\phi^{2p+1}  (A \psi + U) \right] dx, \label{conservation-H-expansion}
\end{eqnarray*}
where we have used the stationary equation (\ref{stationary}) for $\phi$
and the symplectic orthogonality (\ref{orthogonality2}).

\begin{remark}
By direct computation, we can verify that
\begin{eqnarray}
\label{var-principle-soliton}
H_s'(E) + E N_s'(E) = 0,
\end{eqnarray}
for any $E$, for which $\phi(\cdot;E) \in H^2(\R)$ exists.
\end{remark}

\section{Stationary normal-form equation}

We shall recover the results of Theorems \ref{theorem-Kirr}(ii) and \ref{theorem-stability}
on the existence and stability of stationary states from the system
of time evolution equations (\ref{time-W-U})--(\ref{system-a-b}).
Theorems \ref{theorem-Kirr} and \ref{theorem-stability}
were originally proved in \cite{KirrPelin} with the Lyapunov--Schmidt decomposition method
that relies on an orthogonal decomposition with respect to the self-adjoint operator $L_+(E)$.
On the other hand, the decomposition used in the derivation of system
(\ref{time-W-U})--(\ref{system-a-b}) relies on the symplectic orthogonality
conditions (\ref{orthogonality2}). Therefore, computations of this section
provide an alternative proof of Theorems \ref{theorem-Kirr}(ii) and \ref{theorem-stability}.

\subsection{Alternative proof of Theorems \ref{theorem-Kirr}(ii) and \ref{theorem-stability}}

We start with the simplification of the system (\ref{time-W-U})--(\ref{system-a-b})
for stationary solutions of the NLS equation (\ref{GP}).
Because of the symplectic orthogonality conditions (\ref{orthogonality2}), we
define the constrained $H^2$-space,
\begin{equation}
\label{constrained-h-2-space}
H^2_E = \left\{ U \in H^2(\mathbb{R}) : \quad \langle \phi(E), U \rangle_{L^2} = \langle \chi(E), U \rangle_{L^2} = 0 \right\},
\end{equation}
where the subscript indicates that the orthogonal projections are $E$-dependent.
The stationary solutions of the system (\ref{time-W-U})--(\ref{system-a-b}) near $E = E_*$
are described by the following theorem.

\begin{theorem}
Assume that $N_s'(E_*) \neq 0$. There exists sufficiently small $\epsilon > 0$
such that for all $|E-E_*|< \epsilon$,
the nonlinear Schr\"{o}dinger equation (\ref{GP}) admits a stationary solution
\begin{equation}
\label{solution-stationary-GP}
\Psi = e^{i t {\cal E}} \left[ \phi(x;E) + A \psi(x;E) + U \right],
\end{equation}
where
\begin{equation}
\label{E-new-definition}
{\cal E} = E - \frac{\langle \partial_E \phi, N_+(A \psi + U,0) \rangle_{L^2}}{\langle \partial_E \phi, \phi + U \rangle_{L^2}},
\end{equation}
while $U \in H^2_E$ and $A \in \R$ are uniquely defined from the implicit equations
\begin{equation}
\label{U-new-definition}
L_+(E) U = G(A,E,U), \quad  \langle \psi(E), G(A,E,U) \rangle_{L^2} =0.
\end{equation}
with
\begin{equation}
\label{LSbifurcation}
G(A,E,U) = \Lambda^2 A \chi + \frac{\langle \partial_E \phi, N_+(A \psi + U,0) \rangle_{L^2}}{\langle \partial_E \phi, \phi + U \rangle_{L^2}} (\phi + A \psi + U) - N_+(A \psi + U,0).
\end{equation}
Moreover, there exist positive constants $C_1$ and $C_2$ such that
\begin{equation}
\label{estimates-u-a}
\| U \|_{H^2} \leq C_1 |E - E_*|, \quad A^2 \leq C_2 |E - E_*|.
\end{equation}
\label{lemma-bifurcation-stat}
\end{theorem}

\begin{proof}
For real-valued stationary solutions, we can set $B = 0$ and $W = 0$
in the system (\ref{time-W-U})--(\ref{system-a-b}), which give $w = 0$ and
$N_-(A \psi + U,0) = 0$. The modulation equations (\ref{projection2})--(\ref{system-a-b})
become degenerate and can be rewritten in the form,
\begin{eqnarray*}
\langle \partial_E \phi, \phi - U \rangle_{L^2} \dot{E} & = & 0, \\
\langle \partial_E \phi, \phi + U \rangle_{L^2} (\dot{\theta} - E) & = &
- \langle \partial_E \phi, N_+(A \psi + U,0) \rangle_{L^2}, \\
\dot{A} + \dot{E} (-\langle \partial_E \chi, U \rangle_{L^2}
+ A \langle \partial_E \psi, \chi \rangle_{L^2}) & = & 0, \\
- \Lambda^2 A + (\dot{\theta} - E)( A \| \psi \|_{L^2}^2
+ \langle \psi, U \rangle_{L^2}) & = & -\langle \psi, N_+(A \psi + U,0) \rangle_{L^2}.
\end{eqnarray*}

Assuming $N_s'(E_*) \neq 0$ and the smallness of $\| U \|_{L^2}$ for small $|E - E_*|$,
we get $\langle \partial_E \phi, \phi - U \rangle_{L^2} \not=0$ for small $|E - E_*|$.
So, from the first equation, it follows that $\dot{E}=0$. Then, from the third equation,
we infer that $\dot{A} =0$. The second equation gives (\ref{E-new-definition}) with
the correspondence $\dot{\theta} = {\cal E}$.
The fourth equation gives $\langle \psi, G(A,E,U) \rangle_{L^2} =0$, where
$G$ is defined by (\ref{LSbifurcation}) and the normalization $\langle \psi, \chi \rangle_{L^2} = 1$
is used.

The time-evolution system (\ref{time-W-U})--(\ref{time-U-W}) implies
that $U$ becomes time-independent and satisfies
the stationary equation $L_+ U = G(A,E,U)$. So, the
system (\ref{U-new-definition}) is verified
and we shall prove the existence and uniqueness of small solutions
of this system satisfying the estimates
(\ref{estimates-u-a}) for small $|E - E_*|$.

Since $L_+(E_*) \psi_* = 0$ and $\psi(E) \to \psi_*$ in $L^2(\R)$ as $E \to E_*$,
the Implicit Function Theorem gives the existence of a unique local map
\begin{equation}
\label{stationary-map}
\R^2 \ni (A,E) \mapsto U \in H^2(\R) \quad \mbox{\rm near} \quad (A,E) = (0,E_*),
\end{equation}
such that $U$ satisfies equation $L_+(E) U = G(A,E,U)$ subject to the constraint
$\langle \chi(E), U \rangle_{L^2} = 0$, provided that $\langle \psi(E), G(A,E,U) \rangle_{L^2} = 0$.
Let $U_{A,E}$ denote this map for small $A$ and $|E - E_*|$. The map is $C^{\infty}$ if
$p \in \mathbb{N}$.

In addition, it follows from equation
\begin{equation}
\label{gen-ker-L-plus}
L_+(E) \partial_E \phi(E) = - \phi(E)
\end{equation}
that $U_{A,E}$ satisfies $\langle \phi(E), U_{A,E} \rangle_{L^2} = 0$ under the constraint
$\langle \partial_E \phi(E), G(A,E,U) \rangle_{L^2} = 0$, which is identically satisfied.
Therefore, $U_{A,E} \in H^2_E$, according to the definition (\ref{constrained-h-2-space}).

We note that $G(A,E,U)$ is quadratic in $A$ as $A \to 0$. Therefore,
we proceed with a near identity transformation for the map (\ref{stationary-map}),
\begin{equation}
\label{near-identity-transform}
U_{A,E} = A^2 \Theta(x;E) + {\cal O}_{H^2}(A^3),
\end{equation}
where $\Theta \in H^2_E$ is a unique solution of the inhomogeneous equation,
\begin{equation}
\label{correction-theta}
L_+(E) \Theta = p(2p+1) \phi^{2p-1} \psi^2 - p (2p+1) \frac{\langle \partial_E \phi, \phi^{2p-1} \psi^2
\rangle_{L^2}}{\langle \partial_E \phi, \phi \rangle_{L^2}} \phi.
\end{equation}

It remains to control the value of $A$ in terms of the small value of $|E - E_*|$.
Substituting (\ref{perturbation-theory}) and (\ref{near-identity-transform}) to equation
$\langle \psi(E), G(A,E,U) \rangle_{L^2} = 0$, we obtain
\begin{equation}
\label{normal-form-stat}
\lambda'(E_*) \| \psi_* \|^2_{L^2} (E - E_*) A - Q A^3 + {\cal O}(A^4, A^2(E - E_*), (E - E_*)^2) = 0,
\end{equation}
where
\begin{eqnarray*}
Q & = & \frac{1}{3} p (2p+1) (2p-1) \langle \psi_*^2, \phi_*^{2p-2} \psi_*^2 \rangle_{L^2}
+ 2 p(2p+1) \langle \psi_*^2, \phi_*^{2p-1} \Theta_* \rangle_{L^2} \\
& \phantom{t} &  \phantom{text} - p(2p+1) \frac{\langle \partial_E \phi_*, \phi_*^{2p-1} \psi_*^2 \rangle_{L^2}}{\langle \partial_E \phi_*, \phi_* \rangle_{L^2}} \| \psi_* \|_{L^2}^2.
\end{eqnarray*}
Therefore, either $A = 0$ (and $U_{A,E} \equiv 0$) or $A$ is a non-zero root of equation (\ref{normal-form-stat})
that satisfies the second estimate (\ref{estimates-u-a}). Thanks to the expansion (\ref{near-identity-transform}),
the first estimate (\ref{estimates-u-a}) is also satisfied. This concludes the proof of the theorem.
\end{proof}

We will now show that the results following from Theorem \ref{lemma-bifurcation-stat} and
the normal form equation (\ref{normal-form-stat}) are equivalent to the
results of Theorems \ref{theorem-Kirr}(ii) and \ref{theorem-stability}.
Using (\ref{gen-ker-L-plus}), we let
$$
\Theta = \theta + p (2p+1) \frac{\langle \partial_E \phi, \phi^{2p-1} \psi^2
\rangle_{L^2}}{\langle \partial_E \phi, \phi \rangle_{L^2}} \partial_E \phi, \quad
\theta = p(2p+1) L_+^{-1}(E) \phi^{2p-1} \psi^2.
$$
On the other hand, expanding (\ref{E-new-definition}) gives
\begin{equation}
\label{E-cal-E}
{\cal E} = E + p(2p+1) A^2 \frac{\langle \partial_E \phi, \phi^{2p-1} \psi^2 \rangle_{L^2}}{\langle \partial_E \phi, \phi \rangle_{L^2}} + {\cal O}(A^3).
\end{equation}

Using (\ref{derivative-result}), we conclude that
the normal-form equation (\ref{normal-form-stat}) is equivalent to equation
\begin{equation}
\label{normal-form-stat-old}
\lambda'(E_*) \| \psi_* \|^2_{L^2} ({\cal E} - E_*) A - {\cal Q} A^3
+ {\cal O}(A^4, A^2({\cal E} - E_*), ({\cal E} - E_*)^2) = 0,
\end{equation}
where ${\cal Q}$ is given by (\ref{bifurcation-coefficient}) and ${\cal E}$ is a renormalized parameter of
the stationary state $\phi$.

We can see from (\ref{normal-form-stat-old}) that besides zero solution $A = 0$
that corresponds to the symmetric state $\phi(x;{\cal E})$ with ${\cal E} = E$,
there are two nonzero solutions that correspond to the asymmetric states,
\begin{equation}
\label{asymmetric-expansion}
\varphi_{\pm}(x;{\cal E}) = \phi(x;E) \pm A \psi(x;E) + A^2 \Theta(x;E) + {\cal O}_{H^2}(A^3),
\end{equation}
where ${\cal E}$ is related to $E$ by the expansion (\ref{E-cal-E}) and
$A$ is a positive root of the stationary normal-form equation (\ref{normal-form-stat-old})
provided that ${\rm sign}\left(({\cal E} - E_*){\cal Q}\right) = -1$ if $\lambda'(E_*) < 0$.
If ${\cal Q} < 0$, the asymmetric states exist for ${\cal E} > E_*$. If ${\cal Q} > 0$,
the asymmetric states exist for ${\cal E} < E_*$. This concludes the
alternative proof of Theorem \ref{theorem-Kirr}(ii).

To establish the alternative proof of Theorem \ref{theorem-stability}, we need the following result.

\begin{lemma}
Let $\varphi_{\pm}$ be the asymmetric states (\ref{asymmetric-expansion})
that exist for ${\rm sign}\left(({\cal E} - E_*){\cal Q}\right) = -1$ near ${\cal E} = E_*$
and define
$$
L_+(A) = -\partial_x^2 + V(x) - (2p+1) \varphi_+^{2p} + {\cal E}(A),
$$
where ${\cal E}(A)$ is given by the expansion (\ref{normal-form-stat-old}).
Then, the second eigenvalue of $L_+(A)$
is positive for ${\cal Q} < 0$ and negative for ${\cal Q} > 0$.\label{lemma-second-eigenvalue}
\end{lemma}

\begin{proof}
Using (\ref{normal-form-stat}) and (\ref{asymmetric-expansion}), we expand
$\varphi_{\pm}$ for small $A$ by
\begin{eqnarray}
\label{perturbation-varphi}
\varphi_{\pm} = \phi_* \pm A \psi_* + A^2 \left( \Theta_* + \frac{Q}{\lambda'(E_*) \| \psi_* \|^2_{L^2}} \partial_E \phi_* \right)
+ {\cal O}_{H^2}(A^3).
\end{eqnarray}

Let $h(A)$ be the eigenfunction of $L_+(A)$ for the second
eigenvalue $\mu(A)$. We have $h(0) = \psi_*$ and $\mu(0) = 0$ is a simple eigenvalue,
so that the analytic perturbation theory for self-adjoint operators applies.
Using (\ref{normal-form-stat-old}) and (\ref{perturbation-varphi}), we compute
$$
L_+'(0) = -2p(2p+1) \phi_*^{2p-1} \psi_*
$$
and
$$
L_+''(0) = -2p(2p+1)(2p-1) \phi_*^{2p-2} \psi_*^2 -
4 p (2p+1) \phi_*^{2p-1} \psi_*  \left( \Theta_* + \frac{Q}{\lambda'(E_*) \| \psi_* \|^2_{L^2}} \partial_E \phi_* \right)
+ \frac{2 {\cal Q}}{\lambda'(E_*) \| \psi_* \|^2_{L^2}}.
$$

Algorithmic computations yield $h'(0) = 2 \theta_*$, $\mu'(0) = 0$, and, after tedious computations,
$$
\mu''(0) = \frac{\langle L_+''(0) \psi_* + 4 L_+'(0) \theta_*, \psi_* \rangle_{L^2}}{\| \psi_*\|^2_{L^2}} = -\frac{4 {\cal Q}}{\| \psi_*\|^2_{L^2}}.
$$
Therefore, $\mu(A) > 0$ for small $A$ if ${\cal Q} < 0$ and $\mu(A) < 0$ for small $A$ if ${\cal Q} > 0$.
\end{proof}

Let $N_s(E) = \| \phi(\cdot;E) \|_{L^2}^2$ and
$N_a({\cal E}) = \| \varphi_{\pm}(\cdot;{\cal E}) \|^2_{L^2}$.
Using equations (\ref{E-cal-E}) and (\ref{normal-form-stat-old}), we expand $N_a({\cal E})$ in the power series,
\begin{eqnarray*}
N_a({\cal E}) & = & N_s(E) + A^2 \| \psi \|^2_{L^2} + {\cal O}(A^4) \\
& = & N_s(E_*) + N_s'(E_*) (E - E_*) + A^2 \| \psi_* \|^2_{L^2} + {\cal O}((E-E_*)^2,(E-E_*)A^2,A^4) \\
& = & N_s(E_*) + {\cal S} ({\cal E} - E_*) + {\cal O}({\cal E}-E_*)^2,
\end{eqnarray*}
where ${\cal S}$ is given by (\ref{represent-0}). Assuming that ${\cal Q} < 0$,
the asymmetric states exist for ${\cal E} > E_*$. If ${\cal S} > 0$, then
$N_a({\cal E})$ increases with ${\cal E}$, whereas if ${\cal S} < 0$,
then $N_a({\cal E})$ decreases with ${\cal E}$.

Orbital stability and instability of asymmetric stationary states follows from the classical
theorem of Grillakis, Shatah, \& Strauss \cite{GSS} because Lemma \ref{lemma-second-eigenvalue}
shows that the operator $L_+(A)$ linearized at $\varphi_{\pm}({\cal E})$ has one
negative eigenvalue if ${\cal Q} < 0$ and ${\cal E} > E_*$. On the other hand, the symmetric state $\phi(E)$
is unstable for $E > E_*$ by a theorem of Grillakis \cite{Gr}
because Theorem \ref{theorem-Kirr}(i) shows that $L_+(E)$ linearized at $\phi(E)$
has two negative eigenvalues for $E > E_*$. This concludes the alternative proof of Theorem \ref{theorem-stability}.

\subsection{Limit of large separation of potential wells}

The case of large separation of potential wells, when $s \to \infty$ in the double-well
potential (\ref{potential}), gives a good example of explicit computations of
numerical coefficients ${\cal Q}$ and ${\cal S}$. Knowing these numerical coefficients enables
the explicit classification of the stationary states in Theorem \ref{theorem-Kirr}
and \ref{theorem-stability} and verifies the conditions of Theorem \ref{theorem-main}.
The following theorem gives the asymptotic result when $s \to \infty$.

\begin{theorem} Let $V$ be given by \eqref{potential}. There exists sufficiently large $s_0 >0$
such that for all $s > s_0$, we have
$N_s'(E_*) >0$, $\lambda'(E_*) <0$, $\mathcal{Q}<0$, and
\begin{equation*}
\mathcal{S} >0, \ \text{if} \quad p < p_*, \quad \text{and} \quad \mathcal{S} <0 \ \text{if} \quad \ p < p_*,
\end{equation*}
where $p_*$ is the positive root of the equation $1 + 3p -p^2 =0$. \label{theorem-large-separation}
\end{theorem}

\begin{proof}
We recall from \cite{KirrPelin} that
$$
E_* \to E_0 \quad \mbox{\rm and} \quad \psi_*^2 \to C_* \phi_*^2 \;\; \mbox{\rm in} \;\; L^{\infty}(\R)
\quad \mbox{\rm as} \quad s \to \infty,
$$
where $C_* > 0$ is a normalization constant. Using the exact identity
$$
L_+^{-1}(E) \phi^{2p+1} = -\frac{1}{2p} \phi,
$$
we can hence simplify the expression (\ref{bifurcation-coefficient}) to the form,
\begin{eqnarray}
\label{asymptotic-1}
{\cal Q} \to -\frac{4}{3} p (p + 1) (2p + 1) C_*^2 \| \phi_* \|^{2p+2}_{L^{2p+2}} \quad \mbox{\rm as} \quad s \to \infty.
\end{eqnarray}
Therefore, ${\cal Q} < 0$ for any $p \in \mathbb{N}$ if $s$ is sufficiently large.

Because $E_* \to E_0$, we can approximate $\phi_*$ and $\partial_E \phi_*$ using the small-amplitude expansion
for the symmetric states of the stationary equation (\ref{stationary}),
$$
\phi(x;E) = a \phi_0(x) + {\cal O}(a^{1+2p}), \quad E = E_0 + a^{2p} \frac{\| \phi_0 \|_{L^{2p+2}}^{2p+2}}{\| \phi_0 \|^2_{L^2}} + {\cal O}(a^{4p}),
$$
where $\phi_0 \in H^2(\R)$ is the eigenfunction of the operator $L_0 = -\partial_x^2 + V(x)$
for the lowest eigenvalue $-E_0$ and $a \in \R$ is a small parameter of the expansion.
As a result, we obtain
\begin{eqnarray}
\label{asymptotic-2}
\langle \partial_E \phi_*, \phi_*^{2p+1} \rangle_{L^2} \to \frac{1}{2p} \| \phi_* \|^{2}_{L^{2}},
\quad \mbox{\rm as} \quad s \to \infty.
\end{eqnarray}
and
\begin{eqnarray}
\label{asymptotic-23}
N_s'(E_*) = 2\langle \partial_E \phi_*, \phi_* \rangle_{L^2} \to \frac{1}{p} \frac{\| \phi_* \|^4_{L^2}}{\| \phi_* \|^{2p+2}_{L^{2p+2}}}, \quad \mbox{\rm as} \quad s \to \infty.
\end{eqnarray}
In addition, it follows from (\ref{derivative-result}) and (\ref{asymptotic-2}) that
\begin{eqnarray}
\label{asymptotic-2-3}
\lambda'(E_*) \to -2p \quad \mbox{\rm as} \quad s \to \infty.
\end{eqnarray}
Therefore, $N_s'(E_*) > 0$ and $\lambda'(E_*) < 0$ as $s \to \infty$.

Substituting (\ref{asymptotic-1})--(\ref{asymptotic-2-3}) into the expression (\ref{represent-0}), we obtain
\begin{equation}
\label{represent-1}
{\cal S} \to \frac{(1+3p-p^2) \| \phi_* \|^4_{L^2}}{p(1+p)(1+2p) \| \phi_* \|^{2p+2}_{L^{2p+2}}} \quad \mbox{\rm as} \quad s \to \infty.
\end{equation}
Therefore, ${\cal S} > 0$ for $p < p_*$ and ${\cal S} < 0$ for $p > p_*$,
where $p_*$ is given by the positive root of $1 + 3p - p^2 = 0$, that is, by (\ref{threshold}).
\end{proof}

\section{Time-dependent normal form equations}

We rewrite in the abstract form the modulation equations
(\ref{system-theta-E})--(\ref{system-a-b}) for $(\theta,E,A,B)$,
\begin{eqnarray}
\label{system-finite-dimen}
\left\{ \begin{array}{l}
\dot{\theta} - E = R_{\theta}(E,A,B,U,W), \\
\dot{E} = R_E(E,A,B,U,W), \\
\dot{A} - B = R_A(E,A,B,U,W), \\
\dot{B} - \Lambda^2(E) A = R_B(E,A,B,U,W),  \end{array} \right.
\end{eqnarray}
and the system (\ref{time-W-U})--(\ref{time-U-W})
for the remainder terms $(U,W)$,
\begin{eqnarray}\label{system-infinite-dimen}
\left\{ \begin{array}{l} U_t = L_-(E) W + R_U(E,A,B,U,W), \\
- W_t = L_+(E) U +R_W(E,A,B,U,W), \end{array} \right.
\end{eqnarray}
where $R_{\theta}$, $R_E$, $R_A$, $R_B$, $R_U$ and $R_W$ are some
functionals on the solution. These functionals can be computed explicitly. Indeed,
it follows from \eqref{system-theta-E} that
\begin{equation}
\label{R-theta-E.def}
\left [\begin{matrix}
R_{E} \\ R_{\theta} \end{matrix} \right] = \mathcal{M}^{-1}
\left[ \begin{array}{cc} \langle \phi, N_-(A \psi + U,B \chi + W) \rangle_{L^2} \\
- \langle \partial_E \phi, N_+(A \psi + U,B \chi + W) \rangle_{L^2} \end{array} \right],
\end{equation}
where matrix $\mathcal{M}$ is given by (\ref{A.def}).
If $\| U \|_{L^2}, \| W \|_{L^2} \ll 1$ and $N_s'(E_*) \neq 0$,
then $\mathcal{M}$ invertible and
\begin{equation}
\label{M.expansion}
\mathcal{M}^{-1} = \langle \partial_E \phi, \phi \rangle_{L^2}^{-1}
                      \left [\begin{matrix} 1 & 0 \\ 0 & 1\end{matrix} \right]
                     + \mathcal{O}(\| U \|_{L^2} + \| W \|_{L^2}).
\end{equation}
On the other hand, equations \eqref{system-b-a} and \eqref{system-a-b} yield
\begin{eqnarray} \nonumber
R_A & =\langle \chi, N_-(A \psi + U,B \chi + W) \rangle_{L^2} + R_{E}
                    ( \langle \partial_E \chi, U \rangle_{L^2} - A \langle \partial_E \psi, \chi \rangle_{L^2}) + \\ \label{RA.def}
               & \quad \quad \quad \quad R_{\theta}( B \| \chi \|_{L^2}^2 + \langle \chi, W \rangle_{L^2}), \\ \nonumber
R_B & = -\langle \psi, N_+(A \psi + U,B \chi + W) \rangle_{L^2} +
                       R_{E} ( \langle \partial_E \psi, W \rangle_{L^2} - B \langle \partial_E \chi, \psi \rangle_{L^2}) \\ \label{RB.def}
               & \quad \quad \quad \quad - R_{\theta}( A \| \psi \|_{L^2}^2 + \langle \psi, U \rangle_{L^2}).
\end{eqnarray}
Now, the system (\ref{time-W-U})--(\ref{time-U-W}) simplified with
the modulation equations (\ref{system-finite-dimen}) yields
\begin{equation} \label{RUW.def}
\left\{
\begin{split}
R_U & =  N_-(A \psi + U,B \chi + W) + R_{\theta} (B \chi + W) - R_{E} (\partial_E \phi + A \partial_E \psi) - R_A \psi, \\
R_W & = N_+(A \psi + U,B \chi + W) + R_{\theta} (\phi + A \psi + U) + B R_{E} \partial_E \chi + R_{B}\chi.
\end{split}
\right.
\end{equation}

\begin{remark}
If $A = B = 0$, then $U = W = 0$ is an invariant solution of
the system (\ref{system-infinite-dimen}), which
give zero values of $R_{\theta}$, $R_E$, $R_A$, and $R_B$ in
the system (\ref{system-finite-dimen}) for any $E$.  \label{remark-2}
\end{remark}

Observation of Remark \ref{remark-2} inspires us to consider the power series expansions
for solutions of the systems (\ref{system-finite-dimen}) and (\ref{system-infinite-dimen}).
Taking into account the spatial symmetry of eigenfunctions, we can see that
$U$, $W$, $R_{\theta}$, and $R_E$ are quadratic with respect to $(A,B)$,
whereas $R_A$ and $R_B$ are cubic with respect to $(A,B)$. Therefore, we write
\begin{eqnarray} \label{system-truncated}
\left\{ \begin{array}{l}
\dot{\theta} - E = C_1(E) A^2 + C_2(E) B^2 + \wt{R}_\theta(E,A,B,\tilde{U},\tilde{W}), \\
\dot{E} = C_3(E) A B + \wt{R}_E(E,A,B,\tilde{U},\tilde{W}), \\
\dot{A} - B = C_4(E) A^2 B + C_5(E) B^3 + \wt{R}_A(E,A,B,\tilde{U},\tilde{W}), \\
\dot{B} - \Lambda^2(E) A = C_6(E) A^3 + C_7(E) A B^2 + \wt{R}_B(E,A,B,\tilde{U},\tilde{W}),  \end{array} \right.
\end{eqnarray}
and
\begin{eqnarray} \label{remainder-terms}
\left\{ \begin{array}{l}
U = A^2 \Theta(x;E) + B^2 \Delta(x;E) + A^3 U_1(x;E) + A B^2 U_2(x;E) + \wt{U}(x,t), \\
W = A B \Gamma(x;E) + A^2 B W_1(x;E) + B^3 W_2(x;E) + \wt{W}(x,t), \end{array} \right.
\end{eqnarray}
where $\wt{R}_{\theta}$, $\wt{R}_{E}$, $\wt{R}_{A}$, and $\wt{R}_{B}$ are the error terms,
whereas $\wt{U}$ and $\wt{W}$ are the remainder terms.

\begin{remark}
$B = 0$ and $\tilde{W} = 0$ yield the invariant reduction
of the stationary system (\ref{system-truncated}) to the algebraic equation,
$$
-\Lambda^2(E) A = C_6(E) A^3 + \wt{R}_B(E,A,0,\tilde{U},0),
$$
from which it follows that $C_6(E_*) = Q$ in notations of Section 3.1.
\end{remark}

Let us first explicitly compute the coefficients $C_1, C_2, \ldots C_7$
and determine the functions $\Theta, \Delta, \ldots, W_2$. We shall then
estimate the error and remainder terms in (\ref{system-truncated})
and (\ref{remainder-terms}) as quadric with respect to $(A,B)$.
Working in a small neighborhood of $(0,0)$ on the phase plane $(A,B)$
and using $|\Delta N|$ as a small parameter, we consider an ellipsoidal region
on the $(A,B)$-plane such that
\begin{equation}
\label{region}
\exists C > 0 : \quad A^2 + |\Delta N|^{-1} B^2 \leq C |\Delta N|.
\end{equation}

Let $T > 0$ be the maximal time until which we consider solutions of the
modulation equations (\ref{system-truncated}) in the domain (\ref{region}).
We assume (and prove in Section 4.4) that there are positive constants $C_0$, $C_1$, and $C_2$
such that
\begin{equation}
\label{region-2}
T \leq C_0 |\Delta N|^{-1/2},
\end{equation}
and
\begin{equation}
\label{region-3}
 |\dot{\theta} - E_*| \leq C_1 |\Delta N|, \quad
|E - E_*| \leq C_2 |\Delta N|.
\end{equation}
The following theorem provides the control of the error
terms of the system (\ref{system-truncated}) and the remainder terms of the
decomposition (\ref{remainder-terms}).

\begin{theorem}
\label{theorem-remainder}
Assume (\ref{region})--(\ref{region-3}).
There exists $\varepsilon > 0$ such that for any $|\Delta N| < \varepsilon$, there
are positive constants $C_1$ and $C_2$ such that
\begin{equation}
\label{estimate-theorem-1}
\sup_{t \in [0,T]} \left( \| \tilde{U}(\cdot,t) \|_{H^1} + \| \tilde{W}(\cdot,t) \|_{H^1} \right) \leq C_1 (\Delta N)^2
\end{equation}
and
\begin{equation}
\label{estimate-theorem-2}
\sup_{t \in [0,T]} \left( |\tilde{R}_{\theta}| + |\tilde{R}_E| +|\tilde{R}_A| +|\tilde{R}_B| \right) \leq C_2 (\Delta N)^2.
\end{equation}
\end{theorem}

The proof of Theorem \ref{theorem-remainder} is given in Sections 4.1 and 4.2.

\subsection{Power series expansions}

For explicit computations, we use the power series expansions (\ref{power-series-u})--(\ref{power-series-w})
and the decompositions (\ref{decomposition2}) and \eqref{remainder-terms} to expand
\begin{eqnarray*}
 N_+ & = & -p(2p+1)\phi^{2p-1}\psi^2A^2 -p\phi^{2p-1}\chi^2B^2 -p(2p+1) \left(
 \frac{2p-1}{3}\phi^{2p-2}\psi^3 + 2\phi^{2p-1}\psi\Theta \right) A^3 \\
& \quad & \phantom{texttext} -p \left( (2p-1)\phi^{2p-2}\psi\chi^2 +2(2p+1)\phi^{2p-1}\psi\Delta +2\phi^{2p-1}\chi\Gamma \right) AB^2 +\wt{N}_+, \\
N_- & = & -2p\phi^{2p-1}\psi\chi AB -p \left( \phi^{2p-2}\chi^3 +2\phi^{2p-1}\chi\Delta \right) B^3 \\
& \quad & \phantom{texttext} -p \left( (2p-1)\phi^{2p-2}\chi\psi^2 +2\phi^{2p-1}\psi\Gamma +2\phi^{2p-1}\chi\Delta
\right) A^2B +\wt{N}_-,
\end{eqnarray*}
where $\wt{N}_+$ and $\wt{N}_-$ are of the form
\begin{equation} \label{N+-.est}
\wt{N}_+, \wt{N}_- = \mathcal{O}((A^2 + B^2)^2 + (A + B)(\wt{U} + \wt{W}) +
\wt{U}^2 + \wt{W}^2).
\end{equation}
From equations \eqref{RUW.def}, \eqref{system-truncated}, and \eqref{remainder-terms}
we have
\begin{equation} \label{RuRw-dec}
\left\{
\begin{split}
R_U  & = f_{1,1}(E) AB + f_{2,1}(E) A^2B + f_{0,3}(E) B^3 + \tilde{F}_U(E,A,B,\tilde{U},\tilde{W}), \\
R_{W} & = g_{2,0}(E) A^2 + g_{0,2}(E) B^2 + g_{3,0}(E) A^3 + g_{1,2}(E) AB^2 +
\tilde{F}_W(E,A,B,\tilde{U},\tilde{W}),
\end{split}
\right.
\end{equation}
where
\begin{eqnarray*}
f_{1,1} & = & -2p\phi^{2p-1}\psi\chi -C_3\partial_E\phi, \\
g_{2,0} & = & -p(2p+1)\phi^{2p-1}\psi^2 + C_1\phi, \\
g_{0,2} & = & -p\phi^{2p-1}\chi^2 + C_2\phi, \\
f_{2,1} & = & C_1\chi -C_3\partial_E\psi - C_4\psi -p(2p-1)\phi^{2p-2}\chi\psi^2 -2p \phi^{2p-1}\psi\Gamma -2p\phi^{2p-1}\chi\Theta, \\
f_{0,3} & = & C_2\chi -C_5\psi -p\phi^{2p-2}\chi^3 -2p\phi^{2p-1}\chi\Delta, \\
g_{1,2} & = & C_2\psi +C_3\partial_E\chi + C_7\chi -p
\left( (2p-1)\phi^{2p-2}\psi\chi^2 +2(2p+1)\phi^{2p-1}\psi\Delta +2\phi^{2p-1}\chi\Gamma \right),\\
g_{3,0} & = & C_1\psi +C_6\chi -p(2p+1)
\left( \frac{1}{3}(2p-1)\phi^{2p-2}\psi^3 +2\phi^{2p-1}\psi\Theta \right),
\end{eqnarray*}
and
\begin{eqnarray*}
\tilde{F}_U & = & (C_1A^2 +C_2B^2)W + (B\chi +W)R_\theta -(\partial_E\phi +A\partial_E\psi)R_E -R_A\psi \\
& \phantom{t} & \quad -2p\phi^{2p-1} \left( A\psi(A^3U_1 +AB^2U_2 +\wt{U}) +B\chi(A^2BW_1 +B^3W_2 +\wt{W}) \right)  \\
& \phantom{t} & \quad -2p\phi^{2p-1}UW -p(2p-1)\phi^{2p-1} \left( (2AU\psi +U^2)(B\chi +W) +A^2\psi^2W \right) \\
& \phantom{t} & \quad -p\phi^{2p-2} \left( 3B^2\chi^2W + 3B\chi W^2 +W^3 \right) + \wt{N}_-, \\
\tilde{F}_W & = & R_\theta(\phi +A\psi +U) + (C_1A^2 + C_2B^2)U + R_B\chi + R_E\partial_E \chi  + \wt{N}_+.
\end{eqnarray*}

Substituting (\ref{remainder-terms}) and \eqref{RuRw-dec} into
the system (\ref{system-infinite-dimen})
and computing the time derivative of $(E,A,B)$ using
the system (\ref{system-truncated}), we obtain
\begin{eqnarray} \label{quadratic-system}
\left\{ \begin{array}{lcl}
L_+(E) \Theta + \Lambda^2(E) \Gamma + g_{2,0}(E) & = & 0, \\
L_+(E) \Delta + \Gamma + g_{0,2}(E) & = & 0, \\
L_-(E) \Gamma - 2 \Theta - 2 \Lambda^2(E) \Delta + f_{1,1}(E) & = & 0, \end{array} \right.
\end{eqnarray}
and
\begin{eqnarray} \label{cubic-system}
\left\{ \begin{array}{lcl}
L_+(E) U_1 + \Lambda^2(E) W_1 + g_{3,0}(E) & = & 0, \\
L_+(E) U_2 + 2W_1 + 3\Lambda^2(E) W_2 + g_{1,2}(E) & = & 0, \\
L_-(E) W_1 - 3U_1 -2\Lambda^2(E) U_2 -  f_{2,1}(E) & = & 0, \\
L_-(E) W_2 - U_2 - f_{0,3}(E) & = & 0, \end{array} \right.
\end{eqnarray}
and
\begin{eqnarray}
\label{system-infinite-dimen-revised}
\left\{ \begin{array}{l} \tilde{U}_t = L_-(E) \tilde{W} + \tilde{R}_U(E,A,B,\tilde{U},\tilde{W}), \\
- \tilde{W}_t = L_+(E) \tilde{U} + \tilde{R}_W(E,A,B,\tilde{U},\tilde{W}). \end{array} \right.
\end{eqnarray}
where
\begin{equation*}
\begin{split}
\wt{R}_U & = \wt{F}_U - (A^2\partial_E\Theta + B^2\partial_E \Delta + A^3\partial_EU_1
+AB^2\partial_EU_2)(C_3AB +R_E)\\
& \quad \quad -(2A\Theta + 3AU_1 +B^2U_2)(C_4A^2B +C_5B^3 +\wt{R}_A)\\
& \quad \quad - (2B\Delta + 2ABU_2)(C_6A^3 +C_7AB^2 +\wt{R}_B), \\
\wt{R}_W & = \wt{F}_W + (AB\partial_E\Gamma + A^2B\partial_EW_1 + B^3\partial_EW_2)(C_3AB +R_E), \\
& \quad \quad + (B\Gamma + 2ABW_1)(C_4A^2B +C_5B^3 +\wt{R}_A)  \\
& \quad \quad + (A\Gamma + A^2W_1 + 3B^2W_2)(C_6A^3 +C_7AB^2 +\wt{R}_B).
\end{split}
\end{equation*}

We shall now introduce two constrained $L^2$ spaces by
\begin{eqnarray}
\label{system-orthogonality-1}
L^2_+ & = & \left\{ U \in L^2(\mathbb{R}) : \quad \langle \phi, U \rangle_{L^2} = \langle \chi, U \rangle_{L^2} = 0 \right\}, \\
\label{system-orthogonality-2}
L^2_- & = & \left\{ W \in L^2(\mathbb{R}) : \quad \langle \partial_E \phi, W \rangle_{L^2} =
\langle \psi, W \rangle_{L^2} = 0 \right\}.
\end{eqnarray}
Note that the orthogonal projections depend on $E$ but we omit this dependence for the notational convenience.
We can also define $H^s_{\pm}(\mathbb{R})$ as constrained $H^s$ spaces for any $s \geq 0$. In notations of
Section 3.1, we have used $H^2_E \equiv H^2_+$.

The following two lemmas describe solutions of the systems (\ref{quadratic-system}) and (\ref{cubic-system}).

\begin{lemma} \label{quadratic}
There exists sufficiently small $\epsilon > 0$ such that for all $|E-E_*| < \epsilon$,
there exist a unique solution $\Theta, \Delta \in H^2_+(\mathbb{R})$ and
$\Gamma \in H^2_-(\mathbb{R})$ of the system \eqref{quadratic-system}. Moreover,
these solutions are even in $x$, $C^2$ in $E$, and satisfy
\begin{equation} \label{L2-es-quadratic}
\exists C > 0 : \quad \|\partial_E^\alpha \Theta \|_{H^2} + \|\partial_E^\alpha \Delta \|_{H^2} +
\|\partial_E^\alpha \Gamma \|_{H^2} \leq C, \quad \alpha =0, 1, 2.
\end{equation}
\end{lemma}

\begin{proof}
To solve the system (\ref{quadratic-system}) near $E = E_*$, we recall that
operators $L_+(E_*)$ and $L_-(E_*)$ are not invertible so that we shall
set $g_{2,0}, g_{0,2} \in L^2_-(\mathbb{R})$ and $f_{1,1} \in L^2_+(\mathbb{R})$
thanks to the symplectic orthogonality. These
constraints set up uniquely the coefficients $C_1$, $C_2$, and $C_3$,
\begin{equation}
\label{C-1-2-3}
C_1 = p (2p+1) \frac{\langle \partial_E \phi, \phi^{2p-1} \psi^2
\rangle_{L^2}}{\langle \partial_E \phi, \phi \rangle_{L^2}}, \quad
C_2 = p \frac{\langle \partial_E \phi, \phi^{2p-1} \chi^2
\rangle_{L^2}}{\langle \partial_E \phi, \phi \rangle_{L^2}}, \quad
C_3 = - 2p \frac{\langle \psi, \phi^{2p} \chi
\rangle_{L^2}}{\langle \partial_E \phi, \phi \rangle_{L^2}}.
\end{equation}
For $E$ close to $E_*$, the existence and uniqueness of the solution
$\Theta, \Delta \in H^2_+(\mathbb{R})$ and $\Gamma \in H^2_-(\mathbb{R})$ of
the system \eqref{quadratic-system} follow from the gap
between zero (or small) eigenvalues of $L_{\pm}(E)$ and the rest of the spectrum
of $L_\pm(E)$ using the fixed point arguments.

Indeed, for any fixed $E$ such that $|E-E_*|$ is sufficiently small,
for any $\Gamma \in L^2_-(\mathbb{R})$, there exists a unique solution
$\Theta_\Gamma, \Delta_{\Gamma} \in H^2_+(\mathbb{R})$ of the system
\begin{equation}
\label{equation-tech-1}
L_+(E) \Theta_{\Gamma} + \Lambda^2(E) \Gamma  =  -g_{2,0}(E), \quad
L_+(E) \Delta_{\Gamma} + \Gamma  =  -g_{0,2}(E).
\end{equation}
Then, for any two $\Gamma_1, \Gamma_2 \in L^2_-(\mathbb{R})$, we have
\begin{equation*}
L_+(E) \left(\Theta_{\Gamma_1} - \Theta_{\Gamma_2} \right) = \Lambda^2(E)
\left( \Gamma_2 -\Gamma_1\right), \quad
L_+(E) \left(\Delta_{\Gamma_1} - \Delta_{\Gamma_2} \right) =
\Gamma_2 -\Gamma_1.
\end{equation*}
So, by using the elliptic regularity, there is a positive constant $C(E)$ such that
\begin{equation*}
\|\Theta_{\Gamma_1} -\Theta_{\Gamma_2} \|_{H^2} \leq C(E) |\Lambda(E)|^2
\|\Gamma_2 -\Gamma_1 \|_{L^2}, \quad \| \Delta_{\Gamma_1} -\Delta_{\Gamma_2} \|_{H^2} \leq C(E)
\|\Gamma_2 -\Gamma_1 \|_{L^2}.
\end{equation*}
We recall that $\Lambda^2(E) \to 0$ as $E \to E_*$.
From these estimates and by applying the fixed point arguments,
we get the existence of $\Gamma \in H^2_-(\mathbb{R})$ which solves the equation
\begin{equation}
\label{equation-tech-2}
L_-(E) \Gamma - 2 \Theta_\Gamma - 2 \Lambda^2(E) \Delta_\Gamma  = -f_{1,1}(E).
\end{equation}

To estimate the $H^2$-norm of the three solutions, we write equations (\ref{equation-tech-1}) as
\[
L(E_*) \Theta_{\Gamma} =  -g_{2,0}(E) - \left( L_+(E) -L_+(E_*) \right) \Theta_{\Gamma} -\Lambda^2(E) \Gamma
\]
and
\[
L(E_*) \Delta_{\Gamma} =  -g_{0,2}(E) - \left( L_+(E) -L_+(E_*) \right) \Delta_{\Gamma} - \Gamma.
\]
Because $\| L_+(E) -L_+(E_*) \|_{L^{\infty}} = {\cal O}(E-E_*)$ and $\Lambda^2(E) = {\cal O}(E-E_*)$
are sufficiently small for $|E - E_*| < \epsilon$, the standard regularity theorem for elliptic equations
implies that there exists an $E$-independent constant $C > 0$ such that
\begin{equation}
\label{equation-tech-3}
\| \Theta_{\Gamma} \|_{H^2} \leq  C \left( 1 + \epsilon \|\Gamma \|_{H^2} \right), \quad
\| \Delta_{\Gamma} \|_{H^2} \leq  C \left( 1 + \|\Gamma \|_{H^2} \right).
\end{equation}
From these estimates and equation (\ref{equation-tech-2}), we get for some constants $C, \tilde{C} > 0$,
\begin{equation*}
\| \Gamma \|_{H^2} \leq C \left( 1 + \|\Theta_{\Gamma} \|_{H^2} +
\epsilon \|\Delta_{\Gamma} \|_{H^2} \right)
\leq \tilde{C} \left( 1 + \epsilon \|\Gamma \|_{H^2} \right).
\end{equation*}
So, if $|E-E_*| < \epsilon$ is small enough, such that
$\tilde{C} \epsilon < 1$, there is an $E$-independent constant $C > 0$ such that
$\|\Gamma \|_{H^2} \leq  C$. From this estimate and estimate (\ref{equation-tech-3}),
we also obtain the $H^2$-estimates of $\Theta$ and $\Delta$.

Because $L_\pm(E)$, $g_{2,0}(E)$, $g_{0,2}(E)$, and $f_{1,1}(E)$ are all $C^2$ in $E$, we see that
$\Theta$, $\Delta$, and $\Gamma$ are all $C^2$ in $E$. By differentiating the system \eqref{quadratic-system} and using the same method as the one we just used, we also obtain the $H^2$-estimates of $\partial_E^\alpha\Theta,
\partial_E^\alpha\Delta$ and $\partial_E^\alpha \Gamma$ with $\alpha =1,2$.
\end{proof}

\begin{lemma} \label{cubic}
There exists sufficiently small $\epsilon > 0$ such that for all $|E-E_*| < \epsilon$,
there exist a unique solution $U_1, U_2 \in H^2_+(\mathbb{R})$ and $W_1, W_2 \in H^2_-(\mathbb{R})$
of the system (\ref{cubic-system}). Moreover, these solutions are odd in $x$, $C^2$ in $E$, and satisfy
\begin{equation}
\label{L2-es-cubic}
\exists C > 0 : \quad  \| \partial_E^\alpha U_j \|_{H^2} + \|\partial_E^\alpha W_j \|_{H^2} \leq C,
\quad j =1,2, \quad \alpha =0,1,2.
\end{equation}
\end{lemma}

\begin{proof}
To solve the system (\ref{cubic-system}) near $E = E_*$, we set $g_{3,0}, g_{1,2} \in L^2_-(\mathbb{R})$
and $f_{2,1}, f_{0,3} \in L^2_+(\mathbb{R})$ thanks to the symplectic orthogonality.
These constraints set up uniquely the coefficients $C_4$, $C_5$, $C_6$, and $C_7$,
\begin{equation} \label{C4-7}
\begin{split}
C_4 & =  - C_3 \langle \partial_E \psi, \chi \rangle_{L^2} + C_1 \| \chi \|^2_{L^2}
- 2 p \langle \chi^2, \phi^{2p-1} \Theta \rangle_{L^2} - p (2p-1)
\langle \chi^2, \phi^{2p-2} \psi^2 \rangle_{L^2}\\
\quad & \quad \quad \quad \quad  - 2p \langle \chi \psi, \phi^{2p-1} \Gamma \rangle_{L^2}, \\
C_5 & =  C_2 \| \chi \|^2_{L^2}
- 2 p \langle \chi^2, \phi^{2p-1} \Delta \rangle_{L^2} - p
\langle \chi^2, \phi^{2p-2} \chi^2 \rangle_{L^2}, \\
C_6 & = - C_1 \| \psi \|^2_{L^2}
+ 2 p (2p+1) \langle \psi^2, \phi^{2p-1} \Theta \rangle_{L^2}
+ \frac{1}{3} p (2p+1) (2p-1) \langle \psi^2, \phi^{2p-2} \psi^2 \rangle_{L^2}, \\
C_7 & =  - C_3 \langle \partial_E \chi, \psi \rangle_{L^2} - C_2 \| \psi \|^2_{L^2}
+ 2 p (2p+1) \langle \psi^2, \phi^{2p-1} \Delta \rangle_{L^2}
+ 2p \langle \psi \chi, \phi^{2p-1} \Gamma \rangle_{L^2} \\
\ & \quad \quad \quad + p (2p-1) \langle \psi^2, \phi^{2p-2} \chi^2 \rangle_{L^2},
\end{split}
\end{equation}
The rest of the proof is similar to that of Lemma \ref{quadratic}.
\end{proof}

To end this section, we estimate the error terms $\tilde{R}_{\theta}$, $\tilde{R}_E$, $\tilde{R}_A$, and $\tilde{R}_B$ of the system \eqref{system-truncated}.

\begin{lemma} There is $C > 0$ such that
\begin{equation} \label{wtUW-eqn}
|\wt{R}_{\theta,E,A,B}| \leq C \left( (A^2 + B^2)^2 + (A + B)(\| \wt{U} \|_{L^2} + \| \wt{W} \|_{L^2}) +
\|\wt{U} \|^2_{L^2} + \|\wt{W} \|_{L^2}^2\right).
\end{equation}
\end{lemma}

\begin{proof}
We recall from \eqref{R-theta-E.def} and \eqref{M.expansion} that
\begin{equation} \label{REthe}
 \left [\begin{matrix}
 R_{E}\\
R_{\theta} \end{matrix} \right] = \left( \langle \partial_E \phi, \phi \rangle_{L^2}^{-1}  +
\mathcal{O}( \|U\|_{L^2} + \| W \|_{L^2}) \right)
\left[ \begin{array}{cc} \langle \phi, N_-(A \psi + U,B \chi + W) \rangle_{L^2} \\
- \langle \partial_E \phi, N_+(A \psi + U,B \chi + W) \rangle_{L^2} \end{array} \right].
\end{equation}
By using the symmetry properties of $\phi$, $\psi$, and $\chi$,
as well as Lemmas \ref{quadratic} and \ref{cubic}, we get
\begin{equation} \label{Rethe-m}
\begin{split}
& \langle \phi, N_-(A \psi + U,B \chi + W) \rangle_{L^2}
        = \langle \partial_E \phi, \phi \rangle_{L^2}C_3(E)AB +  \langle \phi, \wt{N}_- \rangle_{L^2},\\
& -\langle \partial_E \phi, N_+(A \psi + U,B \chi + W) \rangle_{L^2}
   = \langle \partial_E \phi, \phi \rangle_{L^2}[C_1A^2 + C_2B^2 ] - \langle \partial_E \phi, \wt{N}_+ \rangle,
\end{split}
\end{equation}
where $C_1(E)$, $C_2(E)$, and $C_3(E)$ are defined in \eqref{C-1-2-3}. It then follows from
\eqref{system-truncated}, \eqref{N+-.est}, \eqref{REthe}, and \eqref{Rethe-m} that
$\wt{R}_E$ and $\wt{R}_\theta$ satisfy \eqref{wtUW-eqn}.

The computation of the terms $\wt{R}_A$ and $\wt{R}_B$ can be done exactly the same way using \eqref{RA.def}, \eqref{RB.def}, and the above estimates on $\wt{R}_E$ and $\wt{R}_\theta$.
\end{proof}

\subsection{Analysis of the remainder terms}

We shall now consider the remainder terms $\wt{U}$ and $\wt{W}$ satisfying
the system (\ref{system-infinite-dimen-revised}). Let us denote
$\tilde{\bf Z} = (\tilde{U},\tilde{W})$, $\tilde{\bf R} = (\tilde{R}_U,-\tilde{R}_W)$, and
$$
H(E) = \left[ \begin{array}{cc} 0 & L_-(E) \\ -L_+(E) & 0 \end{array} \right].
$$
The system (\ref{system-infinite-dimen-revised}) can be rewritten in the matrix-vector notations as
\begin{eqnarray}
\label{system-remainder}
\partial_t \tilde{\bf Z} = H(E) \wt{\bf Z} + \wt{\bf R}(E,A,B,\tilde{\bf Z}).
\end{eqnarray}

Let $P_c(E)$ be the projection operator associated to the complement
of the four-dimensional subspace spanned by
$$
\left\{ \left[ \begin{array}{c} \phi \\ 0 \end{array} \right], \;\;
\left[ \begin{array}{c} 0 \\ \partial_E \phi \end{array} \right], \;\;
\left[ \begin{array}{c} 0 \\ \psi \end{array} \right], \;\;
\left[ \begin{array}{c} \chi \\ 0 \end{array} \right] \right\}.
$$

We shall single out quartic terms in $(A,B)$ from the nonlinear term of the system (\ref{system-remainder}).
To do so, we expand
\begin{equation}
\label{wt-R-bf}
\wt{\bf R} = \sum_{i+j = 4} {\bf f}_{ij}(E) A^i B^j + \hat{\bf F}(E,A,B,\tilde{\bf Z}),
\end{equation}
where the $C^1$ functions ${\bf f}_{ij}(E) = P_c(E) {\bf f}_{ij}(E) \in H^2(\mathbb{R})$
can be computed explicitly  near $E = E_*$, whereas the function
$\hat{\bf F}(E,A,B,\tilde{\bf Z})$ satisfies the bounds
\begin{equation}
\label{bounds-preliminary}
\exists C_s > 0 : \quad \| \hat{\bf F} \|_{H^s} \leq C_s \left( |A|^5 + |B|^5 +(|A| +|B|) \|\wt{\bf Z} \|_{H^s}
+ \|\wt{\bf Z} \|_{H^s}^2 \right),
\end{equation}
for any $s > \frac{1}{2}$. (Recall here that $H^s(\R)$ is a Banach algebra for any $s > \frac{1}{2}$.)
Without loss of generality, we can work for $s = 1$.

Now, we mirror the decomposition (\ref{wt-R-bf}) and expand $\wt{\bf Z}$ as
\begin{equation}
\label{remainder-terms-revised}
\wt{\bf Z} = \sum_{i+j =4} {\bf z}_{ij}(E) A^i B^j + \hat{\bf Z}(E,A,B),
\end{equation}
where the $C^1$ functions ${\bf z}_{ij}(E) = P_c(E) {\bf z}_{ij}(E) \in H^2(\mathbb{R})$ can be computed explicitly
near $E = E_*$. The new variable satisfies
\begin{eqnarray}
\label{system-remainder-final}
\partial_t \hat{\bf Z} = H(E) \hat{\bf Z} + \hat{\bf R}(E,A,B,\hat{\bf Z}),
\end{eqnarray}
where the residual term $\hat{\bf R}$ is computed from
$\hat{\bf F}$ similar to how $\tilde{\bf R}$ is computed from $\tilde{\bf F}$.
Therefore, the residual term satisfies the bound
\begin{equation}
\label{bounds}
\exists C > 0 : \quad \| \hat{\bf R} \|_{H^1} \leq C \left( |A|^5 + |B|^5 + (|A| + |B|) \| \hat{\bf Z} \|_{H^1}
+ \| \hat{\bf Z} \|^2_{H^1} \right).
\end{equation}

Because $E$ depends on $t$, the spectral projections associated to the linearized operator $H(E)$
are time dependent. Since $E$ is close to $E_*$, we can fix the value $E_*$ before writing
the time evolution problem (\ref{system-remainder-final}) in the Duhamel form.
In other words, we first rewrite (\ref{system-remainder-final}) as
\begin{eqnarray}
\label{system-remainder-new}
\partial_t \hat{\bf Z} = P_c^*(E_*) H(E_*) P_c(E_*) \hat{\bf Z}
+ H_R(E) \hat{\bf Z} + \hat{\bf R}(E,A,B,\hat{\bf Z}),
\end{eqnarray}
where $H_R(E) = P_c^*(E) H(E) P_c(E) - P_c^*(E_*) H(E_*) P_c(E_*)$ is a $2 \times 2$
matrix-valued function. Thanks to $C^1$ smoothness of this function in $E$ and $x$,
it enjoys the bound
\begin{equation}
\label{bounds-linear}
\exists C > 0 : \quad \| H_R(E) \|_{C^1} \leq C |E - E_*|.
\end{equation}

Then we use the Duhamel principle and write
\begin{eqnarray}
\nonumber
& \phantom{t} & \hat{\bf Z}(t) = P_c^*(E_*) e^{t H(E_*)} P_c(E_*) \hat{\bf Z}(0)  \\
\label{system-integral} & + & P_c^*(E_*) \int_0^t  e^{(t-\tau) H(E_*)} P_c(E_*)
\left[ H_R(E(\tau)) \hat{\bf Z}(\tau) + \hat{\bf R}(E(\tau),A(\tau),B(\tau),\hat{\bf Z}(\tau)) \right] d \tau.
\end{eqnarray}

It follows from stability of the symmetric state $\phi(E)$ in Theorem \ref{theorem-stability} that
the operator $H(E_*)$ has zero eigenvalue of algebraic multiplicity {\em four} and the rest of the spectrum
is purely imaginary and bounded away from zero. Therefore, operator $P_c^* (E_*) e^{t H(E_*)} P_c(E_*)$ forms a semi-group from $H^s(\R)$ to $H^s(\R)$ for any $s \geq 0$ and there is $C_s > 0$ such that
\begin{equation}
\label{bound-semigroup}
\| P_c^* (E_*) e^{t H(E_*)} P_c(E_*) \|_{H^s \to H^s} \leq C_s.
\end{equation}

Local existence and uniquness of solutions $\hat{\bf Z}(t)$ of the integral equation (\ref{system-integral})
follows for any $t \in [0,t_0]$, where $t_0 > 0$ is sufficiently small for fixed-point arguments. The solution
can be continued over the maximal existence interval using standard continuation methods. We shall
now use Gronwall's inequality to control $\| \hat{\bf Z}(t) \|_{H^1}$ over $t \in [0,T]$, where $T$ is bounded by (\ref{region-2}) and $(A,B)$ belong to the domain (\ref{region}).

Using (\ref{bounds}),  (\ref{bounds-linear}), (\ref{system-integral}), and (\ref{bound-semigroup}),
we obtain for any $t \in [0,T]$ that
\begin{eqnarray*}
\exists C > 0 : \quad \| \hat{\bf Z}(t) \|_{H^1} & \leq & C \| \hat{\bf Z}(0) \|_{H^1} +
C \int_0^t |E(\tau) - E_*| \| \hat{\bf Z}(\tau)\|_{H^1} d \tau +
C \int_0^t |A(\tau)| \| \hat{\bf Z}(\tau) \|_{H^1} d \tau \\
& \phantom{t} &  \phantom{texttext} + C T |\Delta N|^{5/2}  + C \int_0^t \| \hat{\bf Z}(\tau) \|^2_{H^1} d \tau.
\end{eqnarray*}

By Gronwall's inequality, we have
\begin{eqnarray*}
\sup_{t \in [0,T]} \| \hat{\bf Z}(t) \|_{H^1} & \leq & C \left( \| \hat{\bf Z}(0) \|_{H^1}
+ T |\Delta N|^{5/2} + T \sup_{t \in [0,T]} \| \hat{\bf Z}(t) \|^2_{H^1} \right)
e^{C \int_0^T \left( |E(\tau) - E_*| + |A(\tau)| \right) d \tau}.
\end{eqnarray*}
If $A$, $T$, and $E$ are estimated by (\ref{region}), (\ref{region-2}), and (\ref{region-3}) respectively,
the last exponential term is bounded as $|\Delta N| \to 0$.
Elementary continuation arguments give that if
$\| \hat{\bf Z}(0) \|_{H^1} \leq C_0 (\Delta N)^2$, then
there is $C > 0$ such that
\begin{eqnarray*}
\sup_{t \in [0,T]} \| \hat{\bf Z}(t) \|_{H^1} \leq C (\Delta N)^2, \quad t \in [0,T],
\end{eqnarray*}
or, by virtue of equations (\ref{region}) and (\ref{remainder-terms-revised}),
there is $C > 0$ such that
\begin{equation}
\label{estimate-theorem-3}
\sup_{t \in [0,T]} \| \tilde{\bf Z}(t) \|_{H^1} \leq C (\Delta N)^2, \quad t \in [0,T],
\end{equation}

Bound (\ref{estimate-theorem-3}) provides the proof of the estimate (\ref{estimate-theorem-1}).
The estimate (\ref{estimate-theorem-2}) follows from (\ref{wtUW-eqn}) and (\ref{estimate-theorem-3}).
All together, the proof of Theorem \ref{theorem-remainder} is complete.

\subsection{Conserved quantities}

To complete the proof of Theorem \ref{theorem-main}, we need to show that the trajectories
of the system (\ref{system-truncated}) remain in the domain (\ref{region}) for $t \in [0,T]$
and satisfy the estimates (\ref{region-2}) and (\ref{region-3}).

Estimates (\ref{region-3}) follow from the first two equations of the system (\ref{system-truncated})
in the domain (\ref{region}) under the estimate (\ref{estimate-theorem-2}) on the error terms
and the estimate (\ref{region-2}) on the maximal time $T$. Therefore, we shall only prove
the estimates (\ref{region}) and (\ref{region-2}). To do so, we work with the last two equations
of the system (\ref{system-truncated}) and employ the conserved quantities (\ref{conservation-N}) and (\ref{conservation-H}).

Expanded at the quadratic terms in $(A,B)$, the conserved quantity for ${\cal N}_0$ becomes
$$
{\cal N}_0 = N_s(E) + A^2 \| \psi \|^2_{L^2}+ B^2 \| \chi \|^2_{L^2} + {\cal O}\left(
(A^2 + B^2)^2 + (A + B)(\| \wt{U} \|_{L^2} + \| \wt{W} \|_{L^2}) +
\|\wt{U} \|^2_{L^2} + \|\wt{W} \|_{L^2}^2 \right).
$$
We note that the terms involving $\wt{U}$ and $\wt{W}$ are controlled by the bound (\ref{estimate-theorem-1})
to be of the higher order than the terms involving $A$ and $B$ in the domain (\ref{region}).
To simplify our notations, we shall then rewrite ${\cal N}_0$ simply as
\begin{eqnarray} \label{conservation-N-truncated}
{\cal N}_0 = N_s(E) + A^2 \| \psi \|^2_{L^2}+ B^2 \| \chi \|^2_{L^2} + {\cal O}(A^2 + B^2)^2.
\end{eqnarray}

\begin{remark}
Computing the derivative of (\ref{conservation-N-truncated}) in time and
using system (\ref{system-truncated}) up to the quadratic order, we obtain
$$
N_s'(E) C_3(E) + 2(\|\psi(E) \|^2_{L^2} + \Lambda^2(E)  \| \chi(E) \|^2_{L^2}) = 0,
$$
which is identically satisfied from system (\ref{linearized-system}) thanks to the identity
\begin{eqnarray}
\nonumber
\|\psi \|^2_{L^2} + \Lambda^2  \| \chi \|^2_{L^2} & = & \langle \psi, L_- \chi \rangle_{L^2} -
\langle \chi, L_+ \psi \rangle_{L^2} = \langle \psi, (L_- - L_+) \chi \rangle_{L^2} \\ \label{tech-iden}
& = & 2p \langle \psi, \phi^{2p} \chi \rangle_{L^2} = -\frac{1}{2} N_s'(E) C_3(E).
\end{eqnarray}
\end{remark}

Expanded at the quadratic terms in $(A,B)$, the conserved quantity for
${\cal H}_0$ becomes
\begin{eqnarray}
\nonumber
{\cal H}_0 & = & H_s(E) + A^2 \int_{\R} \left( \psi_x^2 + V \psi^2  - (2p+1) \phi^{2p} \psi^2 \right) dx
+ B^2 \int_{\R} \left( \chi_x^2 + V \chi^2
- \phi^{2p} \chi^2 \right) dx \\
\label{conservation-H-truncated} & \phantom{t} & + {\cal O}(A^2 + B^2)^2,
\end{eqnarray}
Using the system (\ref{linearized-system}) and the normalization $\langle \psi, \chi \rangle_{L^2} = 1$,
we obtain
\begin{eqnarray*}
\int_{\R} \left( \psi_x^2 + V \psi^2 - (2p+1) \phi^{2p} \psi^2 \right) dx & = & - \Lambda^2(E)
- E \| \psi \|^2_{L^2}, \\
\int_{\R} \left( \chi_x^2 + V \chi^2 - \phi^{2p} \chi^2 \right) dx & = & 1 - E \| \chi \|^2_{L^2}.
\end{eqnarray*}
Using (\ref{conservation-N-truncated}), we can further simplify (\ref{conservation-H-truncated}) to the form
\begin{eqnarray}
\label{conservation-H-truncated-new}
{\cal H}_0 = H_s(E) + E (N_s(E) - {\cal N}_0) - \Lambda^2(E) A^2 + B^2  + {\cal O}(A^2 + B^2)^2.
\end{eqnarray}

We can now extend the conserved quantity ${\cal H}_0$ up to the quartic terms
and write it abstractly as
\begin{eqnarray}
\nonumber
{\cal H}_0 & = & H_s(E) + E (N_s(E) - {\cal N}_0) - \Lambda^2(E) A^2 + B^2 \\
\label{conservation-H-truncated-quartic} & \phantom{t} & + \frac{1}{2}
D_1(E) A^4 + D_2(E) A^2 B^2 + \frac{1}{2} D_3(E) B^4 + {\cal O}(A^2 + B^2)^3,
\end{eqnarray}
where $D_1$, $D_2$, and $D_3$ are some coefficients, which can be computed explicitly.
To avoid lengthy computations, we shall compute these coefficients from
the derivative of (\ref{conservation-H-truncated-quartic}) in time
and using the identity (\ref{var-principle-soliton}) and the system (\ref{system-truncated})
up to the quartic order. This procedure yields two relations between
three coefficients $D_1$, $D_2$, and $D_3$,
\begin{eqnarray*}
D_1 + \Lambda^2 D_2 & = & \frac{1}{2} C_3 ( \| \psi \|_{L^2}^2 + \partial_E \Lambda^2) + C_4 \Lambda^2 - C_6, \\
D_2 + \Lambda^2 D_3 & = & \frac{1}{2} C_3 \| \chi \|_{L^2}^2 + C_5 \Lambda^2 - C_7.
\end{eqnarray*}
Since $\Lambda^2(E_*) = 0$, coefficients $D_1(E_*)$ and $D_2(E_*)$ are determined uniquely
from this system. In particular, using (\ref{derivative-result}),
(\ref{C-1-2-3}), and (\ref{tech-iden}), we compute
\begin{eqnarray}
\nonumber
D_1(E_*) & = & - 2 p^2 (2p+1) \frac{\langle \partial_E \phi_*, \phi_*^{2p-1} \psi_*^2 \rangle_{L^2}
\langle \psi_*, \phi_*^{2p} \chi_* \rangle_{L^2}}{\langle \partial_E \phi_*, \phi_* \rangle_{L^2}} - Q \\
\label{represent-3}
& = & -p (2p+1) \frac{\langle \partial_E \phi_*, \phi_*^{2p-1} \psi_*^2 \rangle_{L^2}
\| \psi_* \|^2_{L^2}}{\langle \partial_E \phi_*, \phi_* \rangle_{L^2}} -Q.
\end{eqnarray}

\subsection{Analysis of dynamics as $E \to E_*$}

Hamiltonian system (\ref{system-truncated}) equipped with conserved quantities (\ref{conservation-N-truncated})
and (\ref{conservation-H-truncated-quartic}) is integrable in the sense of the Liouville up to the
error terms controlled by Theorem \ref{theorem-remainder}.
Using the conserved quantity (\ref{conservation-N-truncated}) and the assumption that
$N_s'(E_*) \neq 0$, we can exclude the variable $E$ near $E = E_*$. Then,
using the conserved quantity (\ref{conservation-H-truncated-quartic}), we can
plot the trajectories of the system (\ref{system-truncated}) on the phase plane $(A,B)$ and show
the topological equivalence of the phase portraits of the system (\ref{system-truncated})
to those of the second-order system (\ref{normal-form-equation-time}).

By the assumptions of Theorem \ref{theorem-main}, we have
$\lambda'(E_*) < 0$, $N_s'(E_*) > 0$, and ${\cal Q} < 0$.
As stated in Theorem \ref{theorem-large-separation}, these assumptions are
satisfied for the symmetric potential (\ref{potential}) in the limit of
large separation of wells, that is, as $s \to \infty$.
Hence the asymmetric states $\varphi_{\pm}(x;{\cal E})$ in the expansion (\ref{asymmetric-expansion})
exists for ${\cal E} > E_*$.

We now proceed with the phase plane analysis for the system (\ref{system-truncated}).
We denote $\Delta N = {\cal N}_0 - N_s(E_*)$ and $\Delta H = {\cal H}_0 - H_s(E_*)$. Assuming that
$\Delta N$ is small, we shall work in the domain (\ref{region}) and use the
expansions \eqref{conservation-N-truncated} and (\ref{conservation-H-truncated-quartic}) rewritten
again as
\begin{eqnarray*}
\Delta N = N_s'(E_*) (E - E_*) + A^2 \| \psi_* \|^2_{L^2}+ B^2 \| \chi_* \|^2_{L^2} + {\cal O}((E-E_*)^2,(E-E_*)(A^2+B^2),(A^2+B^2)^4)
\end{eqnarray*}
and
\begin{eqnarray*}
\Delta H & = & H_s(E) - H_s(E_*) + E (N_s(E) - N_s(E_*) - \Delta N) \\
& \phantom{t} & \phantom{text} - \Lambda^2(E) A^2 + B^2  +
\frac{1}{2} D_1(E) A^4 + D_2(E) A^2 B^2 + \frac{1}{2} D_3(E) B^4 + {\cal O}((A^2+B^2)^6), \\
& = & - E \Delta N + \frac{1}{2} N_s'(E_*) (E - E_*)^2 + \lambda'(E_*) \| \psi_* \|^2_{L^2}
(E - E_*) A^2 + B^2 + \frac{1}{2} D_1(E_*) A^4 + D_2(E_*) A^2 B^2  \\
& \phantom{t} & \phantom{text} + \frac{1}{2} D_3(E_*) B^4 +
{\cal O}((E-E_*)^3,(E-E_*)^2(A^2+B^2),(E-E_*)(A^2+B^2)^2,(A^2+B^2)^3),
\end{eqnarray*}
where we have used (\ref{perturbation-theory}) and (\ref{var-principle-soliton}).
The first conserved quantity is useful to eliminate $E$ in the domain (\ref{region}) by
\begin{eqnarray}
\label{delta-N}
N_s'(E_*) (E - E_*) = \Delta N - A^2 \| \psi_* \|^2_{L^2} - B^2 \| \chi_* \|^2_{L^2} + {\cal O}(\Delta N)^2.
\end{eqnarray}
The second conserved quantity can now be written in the form
\begin{eqnarray}
G = \frac{(\Delta N)}{N_s'(E_*)} \lambda'(E_*) \| \psi_* \|^2_{L^2} A^2 + B^2
+ \frac{1}{2} F_1 A^4 + F_2 A^2 B^2 + \frac{1}{2} F_3 B^4 + {\cal O}(\Delta N)^3,
\end{eqnarray}
where
\begin{eqnarray*}
G & = & \Delta H + E_* \Delta N + \frac{(\Delta N)^2}{2 N_s'(E_*)}, \\
F_1 & = & \frac{\| \psi_* \|^4_{L^2} (1 - 2 \lambda'(E_*))}{N_s'(E_*)} + D_1(E_*), \\
F_2 & = & \frac{\| \psi_* \|^2_{L^2} \| \chi_* \|^2_{L^2} (1 - \lambda'(E_*))}{N_s'(E_*)} + D_2(E_*), \\
F_3 & = & \frac{\| \chi_* \|^4_{L^2}}{N_s'(E_*)} + D_3(E_*).
\end{eqnarray*}

Using (\ref{bifurcation-coefficient}), (\ref{represent-0}), (\ref{perturbation-theory}),
and (\ref{represent-3}), we obtain
$$
F_1 = -\frac{{\cal Q} {\cal S}}{N_s'(E_*)}.
$$
In the domain (\ref{region}), where $A^2 = {\cal O}(\Delta N)$ and $B^2 = {\cal O}(\Delta N)^2$, 
$G$ can be rewritten in the form,
\begin{eqnarray}
\label{delta-H}
G = \frac{(\Delta N)}{N_s'(E_*)} \lambda'(E_*) \| \psi_* \|^2_{L^2} A^2 + B^2
- \frac{{\cal Q} {\cal S}}{2 N_s'(E_*)} A^4 + {\cal O}(\Delta N)^3.
\end{eqnarray}
This quantity is time-preserving for any $t \in [0,T]$. Truncation of (\ref{delta-H})
at the error term ${\cal O}(|\Delta N|^3)$ produces the second-order system (\ref{normal-form-equation-time}).

We can now look at four different cases in the dynamics of the system produced
by the conserved quantity (\ref{delta-H}) for any $t \in [0,T]$.
Recall that $N_s'(E_*) > 0$ and ${\cal Q} < 0$ by our assumption.

\subsubsection{Case $\Delta N > 0$ and ${\cal S} > 0$}

It follows from (\ref{delta-H}) that the critical point $(A,B) = (0,0)$ is a saddle point
of $G$ if $\Delta N > 0$ (recall that $\lambda'(E_*) < 0$). The level set $G = 0$
gives a curve on the phase plane $(A,B)$ given by
\begin{equation}
B^2 = \frac{(\Delta N)}{N_s'(E_*)} |\lambda'(E_*)| \| \psi_* \|^2_{L^2} A^2
+ \frac{{\cal Q} {\cal S}}{N_s'(E_*)} A^4 + {\cal O}(\Delta N)^3.
\end{equation}
This curve contains the point $(0,0)$ up to the terms of the order of $(\Delta N)^3$.

If ${\cal Q} < 0$ and ${\cal S} > 0$, the curve $G = 0$ consists of two symmetric loops on
the plane $(A,B)$ that enclose the points $(\pm A_*,0)$, where $G$ is minimal.
An elementary computation shows that
\begin{equation}
\label{nonzero-roots}
A_*^2 = \frac{\lambda'(E_*) \| \psi_* \|^2_{L^2}}{{\cal Q} {\cal S}} (\Delta N) + {\cal O}(\Delta N)^2.
\end{equation}
It is easy to check that expansion (\ref{nonzero-roots}) follows from
the stationary normal-form equation (\ref{normal-form-stat}) if $(E-E_*)$ is eliminated by
the conserved quantity (\ref{delta-N}).

The critical points $(\pm A_*,0)$ are minima of $G$ (center points). Therefore,
they are surrounded by continuous families of periodic orbits on the phase plane $(A,B)$.
Periodic orbits fill the domain enclosed by the two
loops of the level $G = 0$. There are also periodic orbits for $G > 0$ that surround
all three critical points $(\pm A_*,0)$ and $(0,0)$.

Thus, the truncated system (\ref{system-truncated}) approximates the dynamics near
the unstable symmetric state $\phi(x;E)$ and the stable asymmetric states $\varphi_{\pm}(x;{\cal E})$.
The phase portrait is topologically equivalent to the one on Figure \ref{figure-2} (top left).

\subsubsection{Case $\Delta N > 0$ and ${\cal S} < 0$}

The critical point $(A,B) = (0,0)$ is again a saddle point of $G$ if $\Delta N > 0$.
No other critical points of $G$ for small $A$ and $B = 0$ exist if ${\cal S} < 0$. Therefore,
trajectories on the phase plane $(A,B)$ near $(0,0)$ are all hyperbolic and
they leave the neighborhood of $(0,0)$ in a finite time, except for the stable manifolds.
The truncated system (\ref{system-truncated}) in this case approximates the dynamics near
the unstable symmetric state $\phi(x;E)$. The phase portrait
is topologically equivalent to the one on Figure \ref{figure-2} (top right).

\subsubsection{Case $\Delta N < 0$ and ${\cal S} > 0$}

It follows from (\ref{delta-H}) that the critical point $(A,B) = (0,0)$ is a minimum of $G$ if $\Delta N < 0$.
No other critical points of $G$ for small $A$ and $B = 0$ exist if ${\cal S} > 0$. Therefore,
$(0,0)$ is a center point, which is surrounded by a continuous family of periodic orbits on the
phase plane $(A,B)$. The truncated system (\ref{system-truncated}) in this case approximates the dynamics near
the stable symmetric state $\phi(x;E)$. The phase portrait
is topologically equivalent to the one on Figure \ref{figure-2} (bottom left).

\subsubsection{Case $\Delta N < 0$ and ${\cal S} < 0$}

The critical point $(A,B) = (0,0)$ is again a minimum of $G$ if $\Delta N < 0$.
If ${\cal S} < 0$, there exist two symmetric maxima of $G$ at the points $(\pm A_*,0)$, where $A_*^2$ is given by (\ref{nonzero-roots}). The points $(\pm A_*,0)$ are saddle points
and the level curve of $G$ at $A = A_*$ and $B = 0$ prescribes a pair of heteroclinic orbits connecting
$(\pm A_*,0)$. The pair of heteroclinic orbits encloses a continuous family of closed curves surrounding
the center point $(0,0)$ on the phase plane and corresponding to periodic orbits. The truncated system (\ref{system-truncated}) approximates the dynamics near the stable symmetric state $\phi(x;E)$
and the unstable asymmetric states $\varphi_{\pm}(x;{\cal E})$. The phase portrait in this case is topologically equivalent to the one on Figure \ref{figure-2} (bottom right).

\subsection{The end of the proof of Theorem \ref{theorem-main}}

It follows from Section 4.4 that all nontrivial solutions of the system (\ref{system-truncated})
in the domain (\ref{region}) are topologically equivalent to the ones given by the
second-order system (\ref{normal-form-equation-time}).

The estimate (\ref{region-2}) on the maximal time $T > 0$, during which the solutions
remain in the domain (\ref{region}), follows directly from the integration
of the system (\ref{normal-form-equation-time}) over the time $t \in [0,T]$.
The other estimates of Theorem \ref{theorem-main} follows from Theorem \ref{theorem-remainder}.
The proof of Theorem \ref{theorem-main} appears to be complete.

\end{document}